\documentclass[a4paper,10pt]{amsart}

\usepackage[utf8x]{inputenc}
\usepackage[T1]{fontenc}
\usepackage[american]{babel}
\usepackage{xspace}

\usepackage{hyperref}
\usepackage{amssymb,amsmath,amsthm}

\usepackage{todonotes}

\usepackage{enumerate}

\hypersetup{
pdftitle={Fast Exact Algorithms for some Connectivity Problems Parameterized by Clique-Width}, 
pdfauthor={Benjamin Bergougnoux and Mamadou Moustapha Kant\'e}
}

\def\bN{\mathbb{N}}
\def\bF{\mathbb{F}}

\def\cA{\mathcal{A}}
\def\cB{\mathcal{B}}

\def\cD{\mathcal{D}}

\def\cR{\mathcal{R}}
\def\cS{\mathcal{S}}

\def\CG{\mathbf{CG}}

\def\wc{\mathsf{w}}

\def\wrt{\emph{w.r.t.}\xspace}
\def\ie{\emph{i.e.}\xspace}
\def\eg{\emph{e.g.}\xspace}

\newcommand{\pis}{\Pi}
\def\comp#1{\overline{{#1}}}
\def\block#1{{\sharp}\mathrm{block}(#1)}

\def\cwd{\operatorname{cwd}\xspace}

\makeatletter
\def\moverlay{\mathpalette\mov@rlay}
\def\mov@rlay#1#2{\leavevmode\vtop{%
    \baselineskip\z@skip \lineskiplimit-\maxdimen
    \ialign{\hfil$\m@th#1##$\hfil\cr#2\crcr}}}
\newcommand{\charfusion}[3][\mathord]{
  #1{\ifx#1\mathop\vphantom{#2}\fi
    \mathpalette\mov@rlay{#2\cr#3}
  }
  \ifx#1\mathop\expandafter\displaylimits\fi}
\makeatother

\newcommand{\acjoin}{\mathsf{acjoin}}
\newcommand{\join}{\mathsf{join}}

\newcommand{\proj}{\mathsf{proj}}
\newcommand{\acopt}{\mathbf{ac\textrm{-}opt}}
\newcommand{\opt}{\mathbf{opt}}
\newcommand{\rmc}{\mathsf{rmc}}
\newcommand{\reduce}{\mathsf{reduce}}
\newcommand{\acreduce}{\mathsf{ac\textrm{-}reduce}}

\newcommand{\cc}{\mathsf{cc}}
\newcommand{\acy}{\textsf{acyclic}}
\newcommand{\Active}{\textsf{active}}
\newcommand{\duni}{\uplus}

\newcommand{\tricuts}{\mathsf{cuts}}

\newtheorem{theorem}{Theorem}[section]
\newtheorem{lemma}[theorem]{Lemma}
\newtheorem{proposition}[theorem]{Proposition}
\newtheorem{corollary}[theorem]{Corollary}

\theoremstyle{definition}
\newtheorem{definition}[theorem]{Definition}

\newtheorem{fact}[theorem]{Fact}

\renewcommand{\mid}{\, : \, }

\begin{document}

\title[Fast for CWD]{Fast Exact Algorithms for some Connectivity Problems Parameterized by Clique-Width}

\author{Benjamin Bergougnoux}
\author{Mamadou Kant\'e}
\address{Université Clermont Auvergne, LIMOS, CNRS, Aubière, France.}
 \email{\{mamadou.kante,benjamin.bergougnoux\}@uca.fr}

\thanks{This work is supported by French Agency for Research under the GraphEN project (ANR-15-CE40-0009).}



\begin{abstract} Given a clique-width $k$-expression of a graph $G$, we provide $2^{O(k)}\cdot n$ time algorithms for connectivity constraints on locally checkable
  properties such as \textsc{Node-Weighted Steiner Tree}, \textsc{Connected Dominating Set}, or \textsc{Connected Vertex Cover}. We also propose a
  $2^{O(k)}\cdot n$ time algorithm for \textsc{Feedback Vertex Set}. The best running times for all the considered cases were either
  $2^{O(k\cdot \log(k))}\cdot n^{O(1)}$ or worse.
\end{abstract}

\maketitle

\section{Introduction} 

\emph{Tree-width} \cite{RobertsonS86} is probably the most well-studied graph parameter in the graph algorithm community, and particularly by people working in
\emph{Fixed Parameter Tract\-able} (FPT for short) algorithms, due partly to its numerous structural and algorithmic properties \cite{CourcelleE12,DowneyF13}.  For a while,
people used to think that for many connectivity constraints problems, \eg, \textsc{Hamiltonian Cycle}, \textsc{Steiner Tree}, the naive $k^{O(k)}\cdot n^{O(1)}$ time
algorithm, $k$ the tree-width of the input graph, cannot be improved. Indeed, it seems necessary to know the connected components of the partial solutions in order to be
able to extend them and also certify that the given solution is really connected.  But, quite surprisingly, Cygan et al showed in \cite{CyganNPPvRW11} that some of these
connectivity constraints problems admit randomized $2^{O(k)}\cdot n^{O(1)}$ time algorithms.  The first deterministic $2^{O(k)}\cdot n^{O(1)}$ time algorithms for these
problems was due to Bodlaender et al. in \cite{BodlaenderCKN15}. 
First, let us say that $\cS'$ \emph{represents} the set of partial solutions $\cS$ if whenever there is $S\in \cS$ such that $S$ can be completed into an optimum solution, there is $S'\in \cS'$ that can be also completed into an optimum solution. 
Most of the dynamic programming algorithms are based on this notion of representativity, and proposing a fast algorithm is usually reduced to defining the adequate set $\cS'$.  
For instance, one can define, for each node $u$ of a tree-decomposition, a matrix $M_u$ over the binary field where $M_u[F,F']=1$ if and only if the partial solutions $F$ and $F'$ can be ``joined'' into a valid solution of the instance ($F$ is usually intended to be a subset of a solution, subset included in the already processed part of the graph, and $F'$ the remaining part of the solution, that is not yet known). 

The main contribution of \cite{BodlaenderCKN15} was to show that, for some connectivity constraints problems, for each node $u$ of a tree-decomposition of width $k$, the rank of $M_u$ is at most $2^{O(k)}$; moreover, a maximum weighted basis - such that the solutions indexing its rows represent the solutions indexing the rows of $M_u$ - can be computed in time $2^{O(k)}$.

Nevertheless, despite the broad interest on tree-width, only sparse graphs can have bounded tree-width. But, on many dense graph classes, some NP-hard problems admit polynomial time algorithms, and many of these algorithms can be explained by the boundedness of their \emph{clique-width}, a graph parameter introduced by Courcelle and Olariu \cite{CourcelleO00} and that emerges from the theory of graph grammars. 

Clique-width is defined in terms of the following graph operations: (1) addition of a single vertex labeled $i\in \bN$,~ (3) renaming label $i$ into $j$ ($ren_{i\to j}$),~ (3) addition of edges between vertices labeled $i$ and those labeled $j$ ($add_{i,j}$),~ (4) disjoint union ($\oplus$). 
The \emph{clique-width} of a graph is the minimum number of labels needed to construct it, and the expression constructing it is called \emph{$k$-expression}, $k$ the number of used labels.  
Clique-width generalizes tree-width in the sense that if a graph class has bounded tree-width, then it has bounded clique-width \cite{CourcelleO00}, but the converse is false as cliques have clique-width at most 2 and unbounded tree-width.  
Furthermore, clique-width appears also to be of big importance in FPT algorithms \cite{CourcelleE12}. 
While it is still open whether there exists an FPT algorithm to compute an optimal $k$-expression of a given graph, one can ask when clique-width behaves similarly as tree-width.  
It is known that clique-width is far from behaving similarly as tree-width on some well-studied and well-known difficult
problems such as \textsc{Hamiltonicity} \cite{FominGLS14}.  
On the other hand, Bui-Xuan et al. \cite{BuiXuanTV13}, and Ganian et al. \cite{GanianH10,GanianHO11} managed to prove, after more substantial work than for tree-width, that for locally checkable properties and some sparse problems, one can get $2^{k^{O(1)}}\cdot n^{O(1)}$ time algorithms, $k$ the clique-width of the input graph. 
For some of these algorithms, the exponential in their running times are proved to have a linear dependence on the clique-width, while for others only a polynomial dependence is known.

However, nothing is known on connectivity constraints problems, except some special cases such as \textsc{Feedback Vertex Set} which was proved to admit a
$k^{O(k)}\cdot n^{O(1)}$ time algorithm in \cite{BuiXuanSTV13}, provided the graph is given with a $k$-expression.

\subsection*{Our Contributions.} We investigate connectivity constraints on locally checkable properties, such as \textsc{Connected Dominating Set}, or \textsc{Connected Vertex Cover}. All these problems are the connected variant of a problem in the family of problems called \textsc{$(\sigma, \rho)$-Dominating Set} problems. 
 The family of \textsc{$(\sigma, \rho)$-Dominating  Set} problems was introduced in \cite{TelleP97} and studied in graphs of bounded clique-width in \cite{BuiXuanTV13, OumSV13}.  We recall its definition at the end of Section \ref{sec:prelim}. 
 It is not hard to modify the dynamic programming algorithm from \cite{BuiXuanTV13} that computes a minimum $(\sigma,\rho)$-dominating set in order to compute a minimum connected $(\sigma,\rho)$-dominating set in time $k^{O(k)}\cdot n^{O(1)}$, since it suffices to keep track, for each family of partial solutions, the possible partitions of the label classes induced by them. 
 We modify slightly this naive algorithm and prove that one can define representative sets of size $2^{O(k)}$, yielding $2^{O(k)}\cdot n$ time algorithms resolving this family of problems.

We also consider the \textsc{Feedback Vertex Set} problem, which asks to compute a minimum set of vertices to delete so that the resulting graph is acyclic, and propose
similarly a $2^{O(k)}\cdot n$ time algorithm. 
But, the algorithm, even in the same spirit as the one for connected $(\sigma,\rho)$-dominating set, is less trivial since one has to check also the acyclicity, a task that is not trivial when dealing with clique-width operations as a bunch of edges can be added at the same time. 
Indeed, at each step of the dynamic programming algorithm, when dealing with tree-width the number of vertices that have a neighbor in the rest of the graph is bounded, but for clique-width they can be only classified into a bounded number of equivalence classes (with respect to having the same neighborhood in the rest of the graph); these equivalence classes are the \textit{label classes} of the given \textit{$k$-expression}.

In both cases, we use the same rank-based approach as in \cite{BodlaenderCKN15}, but we need to adapt in the \textsc{Feedback Vertex Set}'s case the operations on
partitions to fit with clique-width operations. The main difficulty is to deal with the acyclicity, which the authors of \cite{BodlaenderCKN15} also encountered but
solved by counting the number of edges. In our case, counting the number of edges would yield an $n^{O(k)}$ time algorithm. Let's explain the ideas of the algorithms with
two examples: (1) \textsc{Connected Dominating Set}, asking to compute a minimum connected set $D$ such that each vertex in $V(G)\setminus D$ has a neighbor in $D$ and (2)
\textsc{Maximum Induced Tree}, that consists in computing a maximum induced tree.

Let $G$ be a graph, for which a $k$-expression is given, and let $(X,Y)$ be a bipartition of its vertex set induced by a subexpression of the $k$-expression. 
Let us denote the labels of vertices in $X$ by $\{1,2,\ldots, k\}$. 
It is worth noticing that at the time we are processing the set $X$, may be not all the edges between the vertices of $X$ are known (those edges may be added in the forthcoming $add_{i,j}$ operations). 
To facilitate the steps of the dynamic programming algorithm, we first assume  that either all the edges between the vertices of $X$ labeled $i$ and the vertices of $X$ labeled $j$ are already known, or none of them, for all distinct $i,j$. 
A $k$-expressions fulfilling this constraint can be computed from any $k$-expression in linear time \cite{CourcelleO00}. 

\begin{enumerate}
\item Let $D$ be a connected dominating set of $G$ and let $D_X:=D\cap X$. First, $D$ is not necessarily connected, neither a dominating set of $G[X]$. So, the usual
  dynamic programming algorithm keeps, for each such $D_X$, the pair $(R,R')$ of sequences over $\{0,1\}^k$, where $R_i:=1$ if and only if $D_X$ contains a vertex labeled
  $i$, and $R'_i:=1$ if and only if $X$ has a vertex labeled $i$ not dominated by $D_X$. One first observes that if $D_X$ and $D'_X$ have the same pair of sequences
  $(R,R')$, then $D_X\cup D_Y$ is a dominating set of $G$ if and only if $D_X'\cup D_Y$ is a dominating set.  Therefore, it is sufficient to keep for each pair $(R,R')$
  of sequences in $\{0,1\}^k$ the possible partitions of $\{1,\ldots,k\}$ corresponding, informally, to the connected components of the graphs induced by the $D_X$'s, and for each
  possible partition, the maximum weight among all corresponding $D_X$'s. Notice that the graphs induced by the $D_X$'s are not necessarily induced subgraphs of $G$. One
  easily checks that these tables can be updated without difficulty following the clique-width operations in time $k^{O(k)}\cdot n^{O(1)}$.

  In order to obtain $2^{O(k)}\cdot n$ time algorithms, we modify this algorithm so that the partitions instead of corresponding to the connected components of the
  graphs induced by the $D_X$'s, do correspond to the connected components of the induced subgraphs $G[D_X]$. For doing so, we do not guess the existence of vertices
  labeled $i$ that are not dominated, but rather the existence of a vertex that will dominate the vertices labeled $i$ (if not already dominated). With this modification,
  the steps of our dynamic programming algorithms can be described in terms of the operations on partitions defined in \cite{BodlaenderCKN15}. We can therefore use the same notion
  of representativity in order to reduce the time complexity. 

\item We consider this example because we reduce the computation of a minimum feedback vertex set to that of a maximum tree. We first observe that we cannot use the same trick as in \cite{BodlaenderCKN15} to ensure the acyclicity, that is counting the number of edges induced by the partial solutions. 
Indeed, whenever an $add_{i,j}$  operation is used, many edges can be added at the same time. Hence, counting the edges induced by a partial solution would imply to know, for each partial solution, the number of vertices labeled $i$, for each $i$.
But, this automatically leads to an $n^{O(k)}$ time algorithm. 
We overcome this difficulty by first defining a binary relation $\acy$ on partitions where $\acy(p,q)$ holds whenever there are forests $E$ and $F$, on the same vertex set, such that $E\cup F$ is a forest, and $p$ and $q$ correspond, respectively, to the connected components of $E$ and $F$. 
In a second step, we redefine some of the operations on partitions defined in \cite{BodlaenderCKN15} in order to deal with the acyclicity. These operations are used to describe the steps of the algorithm. They informally help updating the
  partitions after each clique-width operation by detecting partial solutions that may contain cycles.  
  We also define a new notion of representativity,
  \emph{ac-representativity}, where $\cS'$ ac-represents $\cS$ if, whenever there is $S\in \cS$ that can be completed into an acyclic connected set, there is $S'\in \cS'$
  that can be completed into a connected acyclic set.
  We then prove that one can also compute an ac-representative set of size $2^{O(k)}$, assuming the partitions are on
  $\{1,\ldots,k\}$.

  It remains now to describe the steps of the dynamic programming algorithm in terms of the new operations on partitions. First, we are tempted to keep for each forest
  $F$ of $X$, the partition induced by the transitive closure of $\sim$ where $i\sim j$ whenever there is a vertex $x$ labeled $i$, connected in $F$, to a vertex $y$
  labeled $j$. However, this is not sufficient because we may have in a same connected component two vertices labeled $i$, and any forthcoming $add_{i,j}$ operation will
  create a cycle in $F$ if there is a vertex labeled $j$ in $F$. To overcome this difficulty we index our dynamic programming tables with functions $s$ that inform, for
  each $i$, whether there is a vertex labeled $i$ in the partial solutions, and if yes, in exactly $1$ vertex, or in at least $2$ vertices. Indeed, knowing the existence
  of at least two vertices is sufficient to detect some cycles when we encounter an $add_{i,j}$ operation. We need also to make a difference when we have all the vertices
  labeled $i$ in different connected components, or when at least two are in a same connected component. This will allow to detect all the partial solutions $F$ that may
  contain triangles or cycles on $4$ vertices with a forthcoming $add_{i,j}$ operation. Such cycles cannot be detected by the $\acy$ binary relation on partitions since this
  latter does not keep track of the number of vertices in each label. However, the other kinds of cycles are detected through the $\acy$ binary relation.  We refer the reader
  to Section \ref{sec:fvs} for a more detailed description of the algorithm.
\end{enumerate}

One might notice that our algorithms are optimal under the well-know Exponential Time Hypothesis (ETH) \cite{ImpagliazzoPZ98}.  Indeed, from known Karp-reductions, there
are no $2^{o(n)}$ time algorithms for the considered problems. 
Since the clique-width of a graph is always smaller than its number of vertices, it follows that, unless ETH fails, there are no $2^{o(k)}\cdot n^{O(1)}$ time algorithms, for these problems.

The remainder of the paper is organised as follows. 
The next section is devoted to the main notations, and for the definition of clique-width and of the considered problems. 
The notion of \emph{ac-representativity} and the modified operations on partitions are given in Section \ref{sec:framework}. 
We also propose the algorithm for computing an ac-representative set of size $2^{O(|V|)}$, for sets of weighted partitions on a finite set $V$. 
The algorithms for computing a minimum feedback vertex set and connected $(\sigma,\rho)$-dominating sets are given in, respectively, Sections \ref{sec:fvs} and \ref{sec:dom}. 

\section{Preliminaries}\label{sec:prelim}

The size of a set $V$ is denoted by $|V|$ and its power set is denoted by $2^V$. We write $A\setminus B$ for the set difference of $A$ from $B$, and we write $A\duni B$ for the
disjoint union of $A$ and $B$. We often write $x$ to denote the singleton set $\{x\}$.  For a mapping $f:A\to B$, we let $f^{-1}(b):=\{a\in A\mid b= f(a)\}$ for
$b\in B$.  We let $\min (\emptyset):= +\infty$ and $\max(\emptyset):=-\infty$.  We let $[k]:=\{1,\ldots,k\}$. We denote by $\bN$ the set of non-negative integers and by $\bF_2$ the binary field.

\medskip
\paragraph{\bf Partitions} A partition $p$ of a set $V$ is a collection of non-empty subsets of $V$ that are pairwise non-intersecting and such that
$\bigcup_{p_i\in p} p_i = V$; each set in $p$ is called a \emph{block} of $p$. The set of partitions of a finite set $V$ is denoted by $\Pi(V)$, and $(\Pi(V),\sqsubseteq)$ forms a lattice where $p\sqsubseteq q$ if for each block $p_i$ of $p$ there is a block $q_j$ of $q$ with $p_i\subseteq q_j$. The join operation of this lattice is denoted by $\sqcup$. 
For example, we have \[ \{ \{ 1,2\},\{3,4\},\{5\}\} \sqcup \{\{1\},\{2,3\},\{4\},\{5\}\}= \{ \{ 1,2,3,4\}, \{5\}\}. \]
Let $\block{p}$ denote the number of blocks of a partition $p$. Observe that $\emptyset$ is the only partition of the empty set.  A \emph{weighted partition}
is an element of $\Pi(V)\times \bN$ for some finite set $V$. 


For $p\in \Pi(V)$ and $X\subseteq V$, let $p_{\downarrow X}\in \Pi(X)$ be the partition $\{p_i\cap X\mid p_i\in p\}\setminus \{\emptyset\}$, and for $Y\supseteq V$, let $p_{\uparrow Y}\in \Pi(Y)$ be the partition
$p\cup \left(\bigcup_{y\in Y\setminus V} \{\{y\}\}\right)$. 


\medskip
\paragraph{\bf Graphs} Our graph terminology is standard, and we refer to \cite{Diestel05}. The vertex set of a graph $G$ is denoted by $V(G)$ and its edge set by
$E(G)$. An edge between two vertices $x$ and $y$ is denoted by $xy$ (or $yx$). The subgraph of $G$ induced by a subset $X$ of its vertex set is denoted by
$G[X]$, and we write $G\setminus X$ to denote the induced subgraph $G[V(G)\setminus X]$.  The set of vertices that are adjacent to $x$ is denoted by $N_G(x)$, and for
$U\subseteq V(G)$, we let $N_G(U):=\left(\cup_{v\in U}N_G(v)\right)\setminus U$.  For a graph $G$, we denote by $\cc(G)$ the partition $\{V(C)\mid C$ is a connected component of $G\}$ of
$V(G)$.




\medskip
\paragraph{\bf Clique-Width} A \emph{$k$-labeled graph} is a pair $(G,lab_G)$ with $G$ a graph and $lab_G$ a function from $V_G$ to $[k]$, called the \emph{labeling
  function}; each set $lab_G^{-1}(i)$ is called a \emph{label class} and vertices in $lab_G^{-1}(i)$ are called \emph{$i$-vertices}. The notion of clique-width is defined by
Courcelle et al. \cite{CourcelleER93} and is based on the following operations.
\begin{enumerate}
\item Create a graph, denoted by $\mathbf{1}(x)$, with a single vertex $x$ labeled with $1$.
\item For a labeled graph $G$ and distinct labels $i,j \in [k]$, relabel the $i$-vertices of $G$ with $j$ (denoted by $ren_{i\to j}(G)$). Notice that there are no more $i$-vertices in $ren_{i\to j}(G)$. 
\item For a labeled graph $G$ and distinct labels $i,j\in [k]$, add all the non-existent edges between the $i$-vertices and the $j$-vertices  (denoted by $add_{i,j}(G)$).
\item Take the disjoint union of two labeled graphs $G$ and $H$, denoted by $G\oplus H$, with $lab_{G\oplus H}(v):=\begin{cases}
lab_G(v) & \text{if } x\in V(G),\\
lab_H(v) & \text{otherwise.}
\end{cases}$
\end{enumerate}

A \emph{$k$-expression} is a finite well-formed term built with the four operations above. Each $k$-expression $t$ evaluates into a $k$-labeled graph $(val(t),lab(t)$. 
The \emph{clique-width} of a graph $G$, denoted by $\cwd(G)$, is the minimum $k$ such that $G$ is isomorphic to $val(t)$ for some $k$-expression $t$. 
We can assume without loss of generality that any $k$-expression defining a graph $G$ uses $O(n)$ disjoint union operations and $O(nk^2)$ unary operations \cite{CourcelleO00}.


It is worth noticing that from the recursive definition of $k$-expressions, one can compute in time linear in $|t|$ the labeling function $lab(t)$ of $val(t)$, and hence we will always assume that it is given.

To simplify our algorithms, we will use \emph{irredundant $k$-expressions} (see for instance Section \ref{sec:fvs}).

\begin{definition}[Irredundant $k$-expressions \cite{CourcelleO00}]\label{defn:cw-ir} A $k$-expression is \emph{irredundant} if whenever the
  operation $add_{i,j}$ is applied on a graph $G$, there is no edge between an $i$-vertex and a $j$-vertex in $G$.
\end{definition}

It is proved in \cite{CourcelleO00} that any $k$-expression can be transformed in linear time into an irredundant $k$-expression.

\medskip
\paragraph{\bf Considered Connectivity Problems} For all the problems in this article, we consider the weight function to be on the vertices. Observe that putting weights
on the edges would make some of the considered problems such as \textsc{Steiner Tree} NP-hard even on graphs of clique-width two (since cliques have clique-width two). 

A subset $X\subseteq V(G)$ of the vertex set of a graph $G$ is a \emph{feedback vertex set} if $G\setminus X$ is a forest. 
The problem \textsc{Feedback Vertex Set} consists in finding a minimum feedback vertex set. It is not hard to verify that $X$ is a minimum feedback vertex set of $G$ if and only if $G\setminus X$ is a maximum induced forest. 

The problem \textsc{Node-weighted Steiner Tree} asks, given a subset of vertices $K\subseteq V(G)$ called \emph{terminals}, a subset $T$ of minimum weight such that $K\subseteq T \subseteq V(G)$ and $G[T]$ is connected. 

Let $\sigma$ and $\rho$ be two (non-empty) finite or co-finite subsets of $\bN$. We say that a subset $D$ of $V(G)$ \emph{$(\sigma, \rho)$-dominates} a subset $U\subseteq V(G)$ if for
every vertex $u\in U$, $ |N_G(u)\cap D|\in \sigma$ if $u\in D$, otherwise $|N_G(u)\cap D|\in \rho$. 
We say that a set $X$ is a $(\sigma,\rho)$-dominating set (resp. \emph{co-$(\sigma,\rho)$-dominating set}) of a graph $G$, if $V(G)$ is  $(\sigma, \rho)$-dominated by $X$ (resp. $V(G)\setminus X$).
In the connected version, we are moreover asking $X$ to be connected.

Examples of \textsc{Connected (Co-)$(\sigma,\rho)$-Dominating Set} problems are shown in Table \ref{tab:dom}. 
\begin{figure}[h]
\renewcommand\arraystretch{1.2}
\begin{tabular}{|c|c|c|c|}
\hline
$\sigma$ & $\rho$ & Version  &  Standard name  \\\hline
$\bN$ & $\bN^+$ & \emph{Normal} & \textsc{Connected Dominating Set}  \\\hline
$\bN^+$ & $\bN^+$ & \emph{Normal} & \textsc{Connected Total Dominating Set}  \\ \hline
$\{d\}$ & $\bN$ &  \emph{Normal} &\textsc{Connected Induced $d$-Regular Subgraph}  \\ \hline
$\bN$ & $\{1\}$ & \emph{Normal} & \textsc{Connected Perfect Dominating Set}  \\ \hline
$\{0\}$ & $\bN$ & \emph{Co} & \textsc{Connected Vertex Cover}  \\ \hline
\end{tabular} 
\caption{ Examples of \textsc{Connected (Co-)$(\sigma,\rho)$-Dominating Set} problems, $\bN^+:=\bN\setminus\{ 0\}$.}
\label{tab:dom}
\end{figure}

\section{Representing Sets of Acyclic Weighted Partitions by Matrices}\label{sec:framework}

We recall that a \emph{weighted partition} is an element of $\Pi(V)\times \bN$ for some finite set $V$. 
Our algorithms compute a set of weighted partitions $\cA \subseteq \Pi(V)\times \bN$, for each labeled graph $H$ used in the $k$-expression of the given graph $G$ and for every subset $V\subseteq lab_H^{-1}(V(H))$.
Each weighted partition $(p,w) \in \cA\subseteq \pis(V)\times \bN$ is intended to mean the following: there is a solution $S\subseteq V(H)$ of weight $w$ such that $p$ is the transitive closure of the following equivalence relation $\sim$ on $V$ : $i\sim j$ if there exist an $i$-vertex and a $j$-vertex in the same component of $H[S]$.
Moreover, for every label $i$ in $V$, there is at least one $i$-vertex in $S$. 
For each label $i$ in $V$, we expect the $i$-vertices of $S$ to have an additional neighbor in any extension of $S$ into an optimum solution.
This way, for each label $i\in V$, we can consider the $i$-vertices of $S$ as one vertex in terms of connectivity, since they will have a common neighbor in any extension of $S$.
On the other hand, the labels $j \in [k]\setminus V$ such that $S$ contains at least one $j$-vertex are expected to have no additional neighbor in any extension of $S$ into an optimum solution. Consequently, the vertices in $S$ with a label in $[k]\setminus V$ do no longer play a role in the connectivity.
These expectations allow us to represent the connected components of $H[S]$ by $p$.
Our algorithms will guarantee that the weighted partitions computed from $(p,w)$ are computed accordingly to these expectations.

When considering the \textsc{Feedback Vertex Set} problem, as said in the introduction, the trick used in \cite{BodlaenderCKN15} to deal with acyclicity and that consists in counting the number of edges induced by the partial solutions would yield an $n^{O(k)}$  time algorithm in the case of clique-width.  
Since the partial solutions for the \textsc{Feedback Vertex Set} problem are represented by weighted partitions, we need to certify that whenever we join two weighted partitions and keep it as a partial
solution, it does not correspond to a partial solution with cycles. We introduce in the following a notion of acyclicity between two partitions so that we can identify
the joins of partitions which do not produce cycles.

\begin{definition}
  Let $V$ be a finite set. We let $\acy$ be the relation on $\Pi(V)\times \Pi(V)$ where $\acy(p,q)$ holds exactly when $|V|+\block{p\sqcup q}-(\block{p}+\block{q})=0$. 
\end{definition}

Observe that, if $F_p:=(V,E_p)$ and $F_q:=(V,E_q)$ are forests with components $p=\cc(F_p)$ and $q=\cc(F_q)$ respectively, then $\acy(p,q)$ holds if and only if $E_p\cap E_q= \emptyset$ and $(V,E_p\duni E_q)$ is a forest. The following is then quite easy to deduce. 

\begin{fact}\label{fact:par-rep2} Let $V$ be a finite set. For all partitions $p,q,r\in \Pi(V)$,
$$\acy(p,q)\wedge\acy(p\sqcup q,r)\Leftrightarrow \acy(q,r)\wedge\acy(p,q\sqcup r).$$
\end{fact}

\begin{proof} For a partition $p\in \Pi(V)$, let $f(p):=|V|-\block{p}$. One easily checks that $\acy(p,q)$ holds if and only if $f(p\sqcup q)$ equals $f(p)+f(q)$. One can
  therefore deduce, by an easy calculation from this equivalence, that $\acy(p\sqcup q,r)\wedge\acy(p,q)$ is equivalent to saying that $f(p\sqcup q\sqcup r)$ equals
  $f(p)+f(q)+f(r)$. The same statement holds for $\acy(q,r)\wedge\acy(p,q\sqcup r)$.
\end{proof}

By definition of $\acy$ and of $\sqcup$, we can also observe the following.

\begin{fact}\label{fact:join} Let $V$ be a finite set. Let $q\in \Pi(V)$ and let $X\subseteq V$ such that no subset of $X$ is a block of $q$. Then, for each $p\in \Pi(V\setminus X)$, we can observe
  the following equivalences
  \begin{align}
    p_{\uparrow X} \sqcup q = \{V\} & \Longleftrightarrow p\sqcup q_{\downarrow (V\setminus X)} = \{V\setminus X\}& \text{and} \\
    \acy(p_{\uparrow X},q)  & \Longleftrightarrow \acy(p,q_{\downarrow (V\setminus X)}).
  \end{align}
\end{fact}

We modify in this section the operators on weighted partitions defined in \cite{BodlaenderCKN15} in order to express our dynamic programming algorithms in terms of these
operators, and also to deal with acyclicity.  Let $V$ be a finite set. First, for $\cA\subseteq \pis(V)\times \bN$, let
$$ \rmc(\cA) :=\{(p,w)\in \cA \mid \forall(p,w')\in \cA, w' \leq w)\}.$$ 
This operator, defined in \cite{BodlaenderCKN15}, is used to remove all the partial solutions whose weights are not maximum \wrt to their partitions.

\medskip

\paragraph{\bf Ac-Join.} Let $V'$ be a finite set. For $\cA\subseteq \pis(V)\times \bN$ and $\cB\subseteq \pis(V')\times \bN$, we define $\acjoin(\cA,\cB)\subseteq
\pis(V\cup V')\times \bN$ as 
{\small  
\begin{align*}
  \acjoin(\cA,\cB)&:=\{(p_{\uparrow V'}\sqcup q_{\uparrow V},w_1+w_2)\mid 
                     (p,w_1)\in \cA, (q,w_2)\in \cB
                    \textrm{ and }\acy(p_{\uparrow V'},q_{\uparrow V})\} .
\end{align*}}
This operator is more or less the same as the one in \cite{BodlaenderCKN15}, except that we incorporate the acyclicity condition. It is used to construct partial
solutions while guaranteeing the acyclicity. 

\medskip
 \paragraph{\bf Project.} For $X\subseteq V$ and $\cA\subseteq \pis(V)\times \bN$, let $\proj(\cA,X)\subseteq \pis(V\setminus X)\times \bN$ be
\begin{align*}
\proj(\cA,X) &:=\{(p_{\downarrow (V\setminus X)},w)\mid (p,w)\in \cA \textrm{ and } \forall p_i\in p, (p_i\setminus X) \neq \emptyset\}.
\end{align*}

This operator considers all the partitions such that no block is completely contained in $X$, and then remove $X$ from those partitions. We index our dynamic programming
tables with functions that informs on the label classes playing a role in the connectivity of partial solutions, and this operator is used to remove from the partitions the label classes that are required to no longer play a role in the connectivity of the partial solutions. 
If a partition has a block fully contained in $X$, it means that this block will remain disconnected in the future steps of our dynamic
programming algorithm, and that is why we remove such partitions (besides those with cycles).

\bigskip

One needs to perform the above operations efficiently, and this is guaranteed by the following, which assumes that $\log(|\cA|) \leq |V|^{O(1)}$ for each
$\cA\subseteq \pis(V)\times \bN$ (this can be established by applying the operator $\rmc$).

\begin{proposition}[Folklore]\label{prop:op} The operator $\acjoin$ can be performed in time $|\cA|\cdot |\cB|\cdot |V\cup V'|^{O(1)}$ and the size of its output is upper-bounded by $|\cA|\cdot |\cB|$. The operators $\rmc$ and $\proj$ can be performed in time $|\cA|\cdot |V|^{O(1)}$, and the sizes
  of their outputs are upper-bounded by $|\cA|$.
\end{proposition}

We now define the notion of representative sets of weighted partitions which is the same as the one in \cite{BodlaenderCKN15}, except that we need to incorporate the acyclicity condition as for the $\acjoin$ operator above.

\begin{definition}\label{defn:par-rep} Let $V$ be a finite set and let $\cA\subseteq \pis(V)\times \bN$. For $q\in \pis(V)$, let 
  {$$ \acopt(\cA,q) := \max\{w\mid (p,w)\in \cA,p\sqcup q = \{V\} \textrm{ and } \acy(p,q) \}.$$}
  A set of weighted partitions $\cA'\subseteq \pis(V)\times \bN$ \emph{ac-represents} $\cA$ if for each $q\in \pis(V)$, it holds that $\acopt(\cA,q)=\acopt(\cA',q)$.

  Let $Z$ and $V'$ be two finite sets.  A function $f: 2^{\pis(V)\times \bN}\times Z\to 2^{\pis(V')\times \bN}$ is said to \emph{preserve ac-representation} if for each
  $\cA,\cA'\subseteq \pis(V)\times \bN$ and $z\in Z$, it holds that $f(\cA',z)$ ac-represents $f(\cA,z)$ whenever $\cA'$ ac-represents $\cA$.
\end{definition}

At each step of our algorithm, we will compute a small set $\cS'$ that ac-represents the set $\cS$ containing all the partial solutions.
In order to prove that we compute an ac-representative set of $\cS$, we show that $\cS=f(\cR_1,\dots,\cR_t)$ with $f$ a composition of functions that preserve ac-representation, and $\cR_1,\dots,\cR_t$ the sets of partials solutions associated with the previous steps of the algorithm.
To compute $\cS'$, it is sufficient to compute $f(\cR'_1,\dots,\cR_t')$, where each $\cR_i'$ is an ac-representative set of $\cR_i$.
The following lemma guarantees that the operators we use preserve ac-representation.

\begin{lemma}\label{lem:par-rep} The union of two sets in $2^{\Pi(V)\times \bN}$ and the operators $\rmc$, $\proj$, and $\acjoin$ preserve ac-representation.
\end{lemma}

\begin{proof}  Let $V$ be a finite set and let $\cA$ and $\cA'$ be two subsets of $\pis(V)\times \bN$.
	The proof for the union follows directly from the definition of $\acopt$.

  \medskip \subparagraph{\small \bf Rmc}  Let $q\in \pis(V)$. By the definition of $\rmc$, whenever $(p,w)\in \cA$ is such that
  $p\sqcup q=\{V\}$, $\acy(p,q)$ and $\acopt(\cA,q)=w$, then $(p,w)\in \rmc(\cA)$, otherwise there would exist $(p,w')\in \cA$ with
  $w'>w$ which would contradict $w=\acopt(\cA,q)$. Therefore, $\acopt(\rmc(\cA),q) = \acopt(\cA,q)$. We can then conclude that if $\cA'$ ac-represents $\cA$, it holds that $\rmc(\cA')$ ac-represents $\rmc(\cA)$.

  \medskip \subparagraph{\small \bf Projections.} Because $\proj(\cA,X)=\proj(\proj(\cA,x),X\setminus \{x\})$ for all $X\subseteq V$ and $x\in X$, we can assume that $X=\{x\}$.  
  Let  $q\in \pis(V\setminus \{x\})$. For every $(p,w)\in \cA$, if $\{x\}\in p$, then $p\sqcup q_{\uparrow x}\neq \{V  \}$, and $(p_{\downarrow V\setminus x},w)\notin\proj(\cA,\{x\})$. Otherwise, $(p_{\downarrow V\setminus x},w)\in\proj(\cA,\{x\})$,  and by Fact \ref{fact:join} we have
  \begin{align*}
 p_{\downarrow V\setminus x} \sqcup q = \{V\setminus \{x\}\} & \Longleftrightarrow p\sqcup q_{\uparrow x} = \{V\} & \textrm{ and}\\
  \acy(p_{\downarrow V\setminus x}, q)  &\Longleftrightarrow \acy(p, q_{\uparrow x}). 
  \end{align*}
  Therefore, we have $\acopt(\proj(\cA,\{x\}),q) = \acopt(\cA,q_{\uparrow x})$.  
  From this equality, we can conclude that $\proj(\cA',\{x\})$ ac-represents $\proj(\cA,\{x\})$, for all  $\cA'\subseteq \cA$ such that $\cA'$ ac-represents $\cA$.

  \medskip \subparagraph{\small \bf Ac-Join.} Let $V'$ be a finite set and let $\cB\subseteq \pis(V')\times \bN$. Let $r\in \pis(V\cup V')$.

	 Observe that for all $(q,w_2)\in \cB$, if a subset of  $V'\setminus V$ is a block of $q_{\uparrow V}\sqcup r$, then for all $p\in \Pi(V)$, we have $p_{\uparrow V'}\sqcup q_{\uparrow V} \sqcup r \neq \{ V\cup V'\}$. 
	 Therefore, $\acopt(\acjoin(\cA,\cB),r)=\acopt(\acjoin(\cA,\cB'),r)$ where $\cB'$ is the set of all $(q,w)\in \cB$ such that no subset of $V'\setminus V$ is a block of $q_{\uparrow V}\sqcup r$.
	 
 By definition, $\acopt(\acjoin(\cA,\cB'),r)$ equals {\small 
 	\begin{align*}  \max\{w_1+w_2 \mid (p,w_1)\in \cA \wedge   (q,w_2)&\in \cB'\wedge( p_{\uparrow V'}\sqcup  q_{\uparrow V} \sqcup r) = \{V\cup V'\}  \\ 
      &\wedge \acy( p_{\uparrow V'},q_{\uparrow V}) \wedge \acy(p_{\uparrow V'}\sqcup q_{\uparrow V},r) \}.
        \end{align*}}

  By Fact \ref{fact:par-rep2}, $\acopt(\acjoin(\cA,\cB'),r)$ is then equal to
   {\small \begin{align*} 
   	 \max\{w_1+w_2 \mid (p,w_1)\in \cA \wedge
  (q,w_2)&\in \cB' \wedge (p_{\uparrow V'}\sqcup q_{\uparrow V} \sqcup r) = \{V\cup V'\} \\ 
  & \wedge \acy(q_{\uparrow V},r)  \wedge \acy(p_{\uparrow V'}, q_{\uparrow V}\sqcup r) \}.
    \end{align*}}

  We deduce, by Fact \ref{fact:join} and the definition of $\cB'$, that $\acopt(\acjoin(\cA,\cB'),r)$ equals 
  {\small \begin{align*} 
  	 \max\{w_1+w_2 \mid (p,w_1)\in \cA \wedge  (q,w_2)&\in \cB'\wedge (p\sqcup (q_{\uparrow V} \sqcup r)_{\downarrow V}) = \{V\}\\
       &  \wedge \acy(q_{\uparrow V},r)  \wedge  \acy(p,(q_{\uparrow V}\sqcup r)_{\downarrow V}) \}.
\end{align*}}
Therefore, we can conclude that $\acopt(\acjoin(\cA,\cB),r)$ equals 
{\small \begin{align*}  \max \{\acopt(w_2 + \cA,(q_{\uparrow V} \sqcup
r)_{\downarrow V}) \mid
(q,w_2)\in \cB' \wedge \acy(q_{\uparrow V},r) \}.
\end{align*}}
Therefore, if $\cA'$ ac-represents $\cA$, then we can conclude that $\acopt(\acjoin(\cA,\cB),r)$ equals $\acopt(\acjoin(\cA',\cB),r)$. As this statement is true for all $r\in \Pi(V\cup V')$, we can conclude that $\acjoin(\cA',\cB)$ ac-represents $\acjoin(\cA,\cB)$ whenever $\cA'$ ac-represents $\cA$.
Symmetrically, we deduce that $\acjoin(\cA,\cB^\star)$ ac-represents $\acjoin(\cA,\cB)$ whenever $\cB^\star$ ac-represents $\cB$.
\end{proof}

In the remaining, we will prove that, for every set $\cA\subseteq \pis(V)\times \bN$, we can find, in time $|\cA|\cdot 2^{O(|V|)}$, a subset $\cA'\subseteq \cA$ of size at
most $|V|\cdot 2^{|V|}$ that ac-represents $\cA$. As in \cite{BodlaenderCKN15}, we will encode the ac-representativity by a matrix over $\bF_2$ and show that this one
has rank at most the desired bound. Next, we show that an optimum basis of this matrix ac-represents $\cA$ and such a basis can be computed using the following lemma from
\cite{BodlaenderCKN15}. The constant $\omega$ denotes the matrix multiplication exponent.

\begin{lemma}[\cite{BodlaenderCKN15}]\label{lem:optG} Let $M$ be an $n\times m$-matrix over $\bF_2$ with $m\leq n$ and let $w:\{1,\ldots,n\}\to \bN$ be a weight
  function. Then, one can find a basis of maximum weight of the row space of $M$ in time $O(nm^{\omega-1})$.
\end{lemma}

\begin{theorem}\label{thm:reduce1} There exists an algorithm $\acreduce$ that, given a set of weighted partitions $\cA\subseteq \pis(V)\times \mathbb{N}$, outputs in time
  $|\cA|\cdot 2^{(\omega-1)\cdot|V|}\cdot |V|^{O(1)}$ a subset $\cA'$ of $\cA$ that ac-represents $\cA$ and such that $|\cA'|\leq |V|\cdot 2^{|V|-1}$.
\end{theorem}

\begin{proof} 
  If $V=\emptyset$, then it is enough to return $\cA':=\{ ( \emptyset,w) \}$, where $(\emptyset,w)\in \cA$ and $w$ is maximum because $\emptyset$ is the only partition of the empty set. 
  
  Assume from now that $|V|\ne \emptyset$ (this will ensure that the following definitions exist).
  Let us first define the matrix that encodes the property that the join of two partitions corresponds to a partition arising from a connected solution. 
  
  Let $v_0$ be a fixed element of $V$ and let $\tricuts(V):=\{(V_1,V_2)\mid V_1\duni V_2=V\textrm{ and }v_0\in V_1\}$.  
Let $M$ and $C$ be, respectively, a $(\pis(V),\pis(V))$-matrix and a
$(\pis(V),\tricuts(V))$-matrix, both over $\bF_2$, such that {\small \begin{align*}
       M[p,q] & := \begin{cases} 0 & \textrm{if $p\sqcup q \ne \{V\}$},\\ 1 & \textrm{otherwise}. \end{cases}\\
       C[p,(V_1,V_2)]&:=\begin{cases} 0 & \textrm{if
                                                  $p \not \sqsubseteq (V_1,V_2)$},\\ 1 & \textrm{otherwise}. \end{cases}
\end{align*}}

As in \cite{BodlaenderCKN15,CyganNPPvRW11}, we fix an element $v_0$ to ensure that for all $p\in \Pi(V)$, the number of cuts $(V_1,V_2)$ such that $p\sqcup q\sqsubseteq (V_1,V_2)$ is odd if and only if $p\sqcup q=\{ V\}$. In fact, this number equals $2^{\block{p}-1}$.
This property is used in \cite{BodlaenderCKN15} to prove that $M=C\cdot C^t$. 

  Let $\cA$ be a set of weighted partitions. 
  In order to compute an ac-representative set of $\cA$, we will decompose $\cA$ into a small number of sets $\cA_i$. Then, for each set $\cA_i$, we compute a set $\cA'_i\subseteq \cA_i$ of $\cA_i$ such that the union of the sets $\cA_i'$ ac-represents $\cA$. To compute $\cA_i'$, we use Lemma \ref{lem:optG} to find a maximum basis of the row space of $C$ restricted to $\cA_i$.
  
  For each $0\leq i \leq |V|-1$, let $\cA_i$ be the set $\{p\mid (p,w)\in \cA$ and $|V|-\block{p}=i\}$, and let $C_\cA^i$ be the restriction of $C$ to rows in $\cA_i$.
  Let $\cB_i$ be a basis of the row space of $C_\cA^i$ of maximum weight, where the weights are the weights\footnote{We can assume w.l.o.g. that $\cA=\rmc(\cA)$, and thus for each $p\in \cA$, there is a unique $w\in\bN$ such that $(p,w)\in\cA$.}  of the considered weighted partitions in $\cA$.
  Observe that $|\cB_i|\leq 2^{|V|-1}$ because the rank of $C_\cA^i$ is bounded by $|\tricuts(V)|=2^{|V|-1}$.
   For $p\in\cA_i$, let $\cB_i(p)$ be the subset of $\cB_i$ such that $C[p,(V_1,V_2)]=\sum_{q\in \cB_i(p)}C[q,(V_1,V_2)]$ for all $(V_1,V_2)\in\tricuts(V)$.
   Let $\cA_i'$ be the subset of $\cA$ corresponding to the rows in $\cB_i$, and let $\cA':=\cA_0'\uplus\dots\uplus\cA_{|V|-1}'$. Notice that $|\cA'|\leq |V|\cdot 2^{|V|-1}$.
   
   Since $C_\cA^i$ has $|\cA_i|$ rows and $2^{|V|-1}$ columns, $\cA_i$ is computable in time $|\cA_i|\cdot 2^{|V|-1} \cdot |V|^{O(1)}$.
   By Lemma \ref{lem:optG}, we can compute $\cB_i$ in time $|\cA_i|\cdot 2^{(\omega -1)\cdot |V|}\cdot |V|^{O(1)}$.
   Hence, we can compute $\cA'$ in time $|\cA|\cdot 2^{(\omega-1)\cdot|V|}\cdot|V|^{O(1)}$ because $\{\cA_0,\ldots, \cA_{|V|-1}\}$ is a partition of $\{ p \mid (p,w )\in \cA\}$.
   
  Let us show that for all $p\in\cA_i$ and $r\in \pis(V)$, if $M[p,r]=1$, then there is $q\in \cB_i(p)$ such that
  $M[q, r]=1$.  Now, from the equality $M=C\cdot C^t$, we
  have { \begin{align*}
                M[p,r] &= \sum\limits_{(V_1,V_2)\in \tricuts(V)}C[p,(V_1,V_2)]\cdot C^t[(V_1,V_2),r] \\
                                   & = \sum\limits_{(V_1,V_2)\in \tricuts(V)}\left(\sum_{q\in \cB_i(p)} C[q,(V_1,V_2)]\right)\cdot C^t[(V_1,V_2),r] \\
                                   & = \sum\limits_{q\in \cB_i(p)} \left(\sum\limits_{(V_1,V_2)\in \tricuts(V)}C[q,(V_1,V_2)]\cdot C^t[(V_1,V_2),r]\right) \\
                                   & = \sum_{q\in \cB_i(p)} M[q,r].
              \end{align*}}
            So, $M[p, r]=1$ if and only if there is an odd number of $q\in \cB_i(p)$ such that $M[q,r]=1$. 

            It remains now to show that $\cA'$ ac-represents $\cA$. Since by construction $\cA'\subseteq \cA$, for each $q\in \pis(V)$, we have
            $\acopt(\cA,q)\geq \acopt(\cA',q)$.  Assume towards a contradiction that $\cA'$ does not ac-represent $\cA$. Thus, there is $q\in \pis(V)$ such that
            $\acopt(\cA,q)>\acopt(\cA',q)$, and hence there is $(p,w)\in \cA\setminus \cA'$ such that $p\sqcup q=\{V\}$, $\acy(p,q)$ holds and $w>\acopt(\cA',q)$. Let
            $i:=|V|-\block{p}$. Hence, $p\in \cA_i$ and there exists $p'\in\cB_i(p)$  such that $M[p',q]=1$ that is $p'\sqcup q=\{V\}$.  
            Let $(p^\star,w^\star)\in \cA_i'$ such that $p^\star\in\cB_i(p)$, $p^\star\sqcup q=\{V\}$ and $w^\star$ is maximum.
            
            Since $(p^\star,w^\star)\in \cA_i'$, we have $|V|-\block{p^\star}=|V|-\block{p}=i$. 
            We can conclude that $\acy(p^\star,q)$ holds because $\acy(p,q)$ holds.
            Indeed, by definition, $\acy(p,q)$ holds if and only if $|V|+\block{p\sqcup q} -( \block{p} +\block{q})=0$.
            Since $p\sqcup q=\{V\}$, we deduce that $\acy(p,q)$ holds if and only if $i=|V|+1 -\block{q}$.
            Because $p^\star\sqcup q=\{V\}$ and $|V|-\block{p^\star}=i$, we can conclude that $\acy(p^\star,q)$ holds. 
            
            Hence, we have $\acopt(\cA',q)\geq w^\star$.
            Since $w>\acopt(\cA',q)$, it must hold that $w>w^\star$. 
            But, $(\cB_i\setminus \{p^\star\})\cup \{p\}$ is also a basis of $C_\cA^i$ since the set of independent row sets of a matrix forms a matroid. Since $w>w^\star$, the weight of $(\cB_i\setminus \{p^\star\})\cup \{p\}$ is strictly greater than the weight of $\cB_i$, yielding a contradiction.
\end{proof}

\section{Feedback Vertex Set}\label{sec:fvs}

We will use the weighted partitions defined in the previous section to represent the partial solutions. At each step of our algorithm, we will ensure that the stored
weighted partitions correspond to acyclic partial solutions. 
However, the framework in the previous section deals only with connected acyclic solutions. 
So, instead of computing a maximum induced forest (the complementary of a minimum feedback vertex set), we will compute a maximum induced tree.
As in \cite{BodlaenderCKN15}, we introduce a hypothetical new vertex, denoted by $v_0$, that is universal and we compute a pair $(F,E_0)$ so that $F$ is a maximum induced forest of $G$, $E_0$ is a subset of edges incident to $v_0$, and $(V(F)\cup \{v_0\}, E(F)\cup E_0)$ is a tree.  For a weight function $\wc$ on the vertices of $G$ and a subset $S\subseteq V(G)$, we denote by $\wc(S):=\sum_{v\in S}{\wc(v)}$. In order now to reduce the sizes of the dynamic programming
tables, we will express the steps of the algorithm in terms of the operators on weighted partitions defined in the previous section.

Let us explain the idea of the algorithm before defining the dynamic programming tables and the steps. Let $H$ be a $k$-labeled graph. We are interested in storing \emph{ac-representative sets} of all
induced forests of $H$ that may produce a solution. If $F$ is an induced forest of $H$, we would like to store the partition $p$ corresponding to the quotient set of the transitive closure of the
relation $\sim$ on $V(F)$ where $x\sim y$ if $x$ and $y$ have the same label or are in the same connected component. If $J\subseteq [k]$ is such that $\bigcup_{x\in V(F)}lab_H(x)=J$, then this is equivalent to
storing the partition $p$ of $J$ where $i$ and $j$ are in the same block if there are, respectively, an $i$-vertex $x$ and a $j$-vertex $y$ in the same connected component of $F$. 

Now, if $H$ is used in a $k$-expression of a $k$-labeled graph $G$, then in the clique-width operations defining $G$ we may add edges between the $i$-vertices and the $j$-vertices of $H$, for some $i,j\in J$. 
Now, this has no effect if there are exactly one $i$-vertex and one $j$-vertex at distance one in $H[F]$, otherwise cycles may be created, \eg, whenever an $i$-vertex and a $j$-vertex are non-adjacent and belong to the same connected component, or the number of $i$-vertices and $j$-vertices are both at least $2$. 
Nevertheless, we are not able to handle all these cases with the operators on weighted partitions. To resolve the situation where an $i$-vertex and a $j$-vertex are already adjacent, we consider \emph{irredundant $k$-expressions}, \ie, whenever an operation $add_{i,j}$ is used there are no edges
between $i$-vertices and $j$-vertices.  For the other cases, we index the dynamic programming tables with functions $s:[k]\to \{\gamma_0,\gamma_1,\gamma_2,\gamma_{-2}\}$ that tells, for each $i \in [k]$, if
the label class $lab_H^{-1}(i)$ does not intersect $F$ ($s(i)=\gamma_0$), or if it does, in one vertex ($s(i)=\gamma_1$), or in at least two ($s(i)\in \{\gamma_2,\gamma_{-2}\}$)
vertices. We have two possible values for label classes intersecting $V(F)$ in at least two vertices because whenever two $i$-vertices belong to the same connected component of $F$, $F$ does not
produce a valid solution once an operation $add_{i,\ell}$ is applied to $H$ with $s(\ell)\neq \gamma_0$. So, if a label class $lab_H^{-1}(i)$ intersects $V(F)$ in at least two vertices, since we do not know whether a clique-width
operation $add_{i,\ell}$ with $s(\ell)\neq\gamma_0$ will be applied to $H$, we guess it in the function $s$, and whenever $s(i)$ equals $\gamma_{-2}$, we throw $p$ when we encounter an $add_{i,\ell}$ operation (with $s(\ell)\neq \gamma_0$), and if
$s(i)=\gamma_2$, we force $F$ to not having two $i$-vertices in the same connected component. 

Nonetheless, even though we are able to detect the partitions corresponding to induced subgraphs with cycles, taking $p$ as the transitive closure of the relation $\sim$
on $[k]$ describe above may detect false cycles. Indeed, let $x_i,x_j$ and $x_\ell,x'_\ell$ be, respectively, an $i$-vertex, a $j$-vertex and two $\ell$-vertices, such
that $x_i$ and $x_\ell$ belong to the same connected component in $F$, and similarly, $x_j$ and $x_\ell'$ in another connected component of $F$. Now, if we apply an
operation $add_{i,j}$ on $H$, we may detect a cycle with $p$ (through the $\acjoin$ operation), which may not exist when for instance there are only one $i$ and one
$j$-vertex in $F$, both in different connected components.  We resolve this case with the functions $s$ indexing the dynamic programming tables by forcing each label $i$ in $s^{-1}(\gamma_2)$ to wait for exactly one clique-width operation $add_{i,t}$ for some $t\in [k]$. We, therefore, translate all the acyclicity tests to the
$\acjoin$ operation. Indeed, the case explained will be no longer a false cycle as $x_i$ and $x_j$ will be adjacent (with the $add_{i,j}$ operation), and we know that
$x_\ell$ and $x_\ell'$ will be connected to some other vertex in $F$, and since $x_i$ is connected to $x_\ell$ and $x_j$ to $x_\ell'$, we have a cycle.  The following
notion of \emph{certificate graph} formalizes this requirement.


\begin{definition}[Certificate graph of a solution]\label{defn:certif} Let $G$ be a $k$-labeled graph, $F$ an induced forest of $G$, $s:[k]\to \{\gamma_0,\gamma_1,\gamma_2,\gamma_{-2}\}$, and $E_0$
  a subset of edges incident to $v_0$. Let $\beta$ be a bijection from $s^{-1}(\gamma_2)$ to a set $V_s^+$ disjoint from $V(G)\cup \{v_0\}$. The \emph{certificate graph of
    $(F,E_0)$ with respect to $s$}, denoted by $\CG(F,E_0,s)$, is the graph $(V(F)\cup V_s^+\cup \{v_0\}, E(F)\cup E_0\cup E_s^+)$ with
\[   E_s^+:= \bigcup_{i \in s^{-1}(\gamma_2)} \{\{v,\beta(i)\} \mid v\in (V(F)\cap lab_G^{-1}(i))\}. \]
\end{definition}

In a certificate graph of $(F,E_0)$, the vertices in $V_s^+$ represent the expected future neighbors of all the vertices in $V(F)\cap lab_G^{-1}(s^{-1}(\gamma_2))$.
For convenience, whenever we refer to a vertex $v_i^+$ of $V_s^+$, we mean the vertex of $V_s^+$ adjacent to the $i$-vertices in $\CG(F,E_0,s)$. See Figure \ref{fig:Fplus} for an example
of a certificate graph. 

\begin{figure}[h]
\centering
\includegraphics[width=0.7\linewidth]{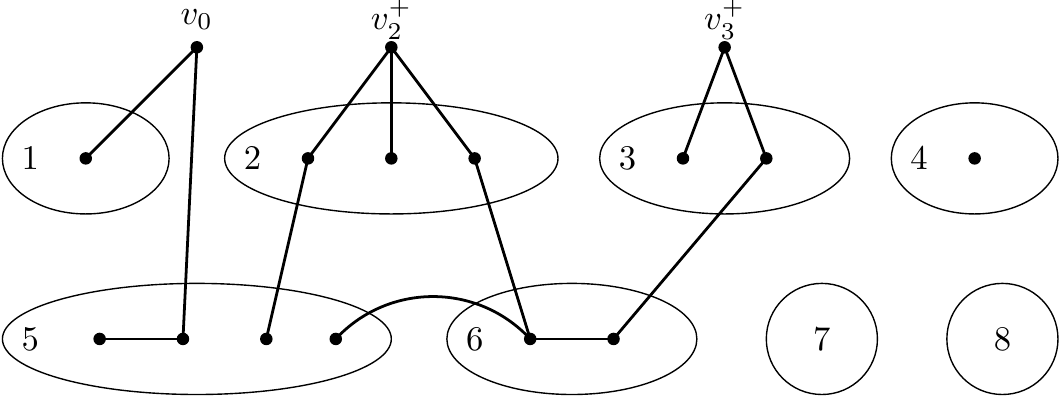}
\caption{Example of a certificate graph; here $p=\{ \{v_0,1\}, \{2,3\}, \{4\} \}$, $s^{-1}(\gamma_0)=\{7,8\}$, $s^{-1}(\gamma_1)=\{1,4\}$, $s^{-1}(\gamma_2)=\{2,3\}$ and
  $s^{-1}(\gamma_{-2})=\{5,6\}$. The set $V_s^+$ is $\{v_2^+,v_3^+\}$ with $v_2^+$ mapped to $2$ and $v_3^+$ mapped to $3$.}
\label{fig:Fplus}
\end{figure}

We are now ready to define the sets of weighted partitions which representatives we manipulate in our dynamic programming tables. 
\begin{definition}[{Weighted partitions in $\cA_G[s]$}]\label{defn:tabfvs} Let $G$ be a $k$-labeled graph and let $s:[k]\to \{\gamma_0,\gamma_1,\gamma_2,\gamma_{-2}\}$ be a total function. The entries of $\cA_G[s]$ are all weighted partitions
  $(p,w)\in \pis(s^{-1}(\{\gamma_1,\gamma_2\})\cup \{v_0\})\times \bN$ such that there exist an induced forest $F$ of $G$ and $E_0\subseteq \{v_0v \mid v\in V(F) \}$ so that $\wc(V(F))=w$, and
\begin{enumerate}
\item The sets $s^{-1}(\gamma_0) = \{i\in [k] \mid |V(F)\cap lab_G^{-1}(i)|=0\}$ and $s^{-1}(\gamma_1) = \{i\in [k] \mid |V(F)\cap lab_G^{-1}(i)|=1\}$.
\item  The certificate graph $\CG(F,E_0,s)$ is a forest.
 
\item Each connected component $C$ of $\CG(F,E_0,s)$ has at least one vertex in $lab_G^{-1}(s^{-1}(\{\gamma_1,\gamma_2\}))\cup \{v_0\}$.

\item The partition $p$ equals $(s^{-1}(\{\gamma_1,\gamma_2\})\cup \{v_0\})/\sim$ where $i\sim j$ if and only if a vertex in $lab_G^{-1}(i)\cap V(F)$ is connected, in $\CG(F,E_0,s)$,
  to a vertex in $lab_G^{-1}(j)\cap V(F)$; we consider $lab_G^{-1}(v_0)=\{v_0\}$.
\end{enumerate}
\end{definition}

\medskip

Conditions (1) and (3) are as explained above, and automatically imply that, for each $i$ in $s^{-1}(\{\gamma_2,\gamma_{-2}\})$, we have $|V(F)\cap lab_G^{-1}(i)|\geq 2$.
Conditions (2) and (4) guarantee that $(V(F)\cup \{v_0\},E(F)\cup E_0)$ can be extended into a tree, if any. They also guarantee that cycles detected through the
$\acjoin$ operation correspond to cycles, and each cycle can be detected with it.

\medskip
In the sequel, we call any triplet $(F,E_0,(p,\wc(F)))$ a \emph{candidate solution} in $\cA_G[s]$ if Condition (4) is satisfied, and if in addition Conditions (1)-(3) are satisfied, we call it a \emph{solution} in $\cA_G[s]$.

\medskip

Given a $k$-labeled graph $G$, the size of a maximum induced tree of $G$ corresponds to the maximum, over all $s:[k]\to \{\gamma_0,\gamma_1,\gamma_2,\gamma_{-2}\}$ with $s^{-1}(\gamma_{2})=\emptyset$, of $\max\{w \mid (\{ s^{-1}(\gamma_{1})\cup\{v_0\} \},w)\in \cA_G[s]\}$. Indeed, by definition, if $(\{ s^{-1}(\gamma_{1})\cup\{v_0\} \},w)$ belongs to $\cA_{G}[s]$, then there exist an induced forest $F$ of $G$ with $\wc(F)=w$ and a set $E_0$ of edges incident to $v_0$ such that $(V(F)\cup \{v_0\},E(F)\cup E_0)$ is a tree.
This follows from the fact that if $s^{-1}(\gamma_{2})=\emptyset$, then we have $\CG(F,E_0,s)=(V(F)\cup \{v_0\}, E(F)\cup E_0)$.

\medskip 

Our algorithm will store, for each $k$-labeled graph $G$ and each function $s:[k]\to \{\gamma_0,\gamma_1,\gamma_2,\gamma_{-2}\}$, an ac-representative set $tab_G[s]$ of
$\cA_G[s]$. We are now ready to give the different steps of the algorithm, depending on the clique-width operations.

\medskip
\paragraph{\bf Computing $tab_G$ for $G=\mathbf{1}(x)$} For $s:\{1\}\to \{\gamma_0,\gamma_1,\gamma_2,\gamma_{-2}\}$, let
\begin{align*}
  tab_G[s] &:=\begin{cases} \{ ( \{\{1,v_0\}\}, \wc(x)), ( \{\{1\},\{v_0\}\}, \wc(x)) \} & \textrm{if $s(1)=\gamma_1$},\\
  \{( \{\{v_0\}\}, 0)\} & \textrm{if $s(1)=\gamma_0$},\\
  \emptyset & \textrm{if $s(1)\in \{\gamma_2,\gamma_{-2}\}$}.
\end{cases}
\end{align*}

Since $|V(G)|=1$, there is no solution intersecting the label class $lab_G^{-1}(1)$ on at least two vertices, and so the set of weighted partitions satisfying Definition
\ref{defn:tabfvs} equals the empty set for $s(1)\in \{\gamma_2,\gamma_{-2}\}$. If $s(1)=\gamma_1$, there are two possibilities, depending on whether $E_0=\emptyset$ or
$E_0=\{xv_0\}$. We can thus conclude that $tab_G[s]=\cA_G[s]$ is correctly computed.

\medskip

\paragraph{\bf Computing $tab_G$ for $G=add_{i,j}(H)$} We can suppose that $H$ is $k$-labeled. Let $s:[k] \to \{\gamma_0,\gamma_1,\gamma_2,\gamma_{-2}\}$.

\begin{enumerate}[(a)]
	\item If $s(i)=\gamma_0$ or $s(j)=\gamma_0$, then let $tab_G[s]:=tab_H[s]$.  
	We just copy all the solutions not intersecting $lab_H^{-1}(i)$ or $lab_H^{-1}(j)$. 
	In the following cases, we assume that $s(i)\ne\gamma_0$ and $s(j)\ne\gamma_0$.
	In this case, we do not need to use the operators $\acreduce$ and $\rmc$ since we update $tab_G[s]$ with one table from $tab_H$.
	
	\item If $s(i)=\gamma_{2}$ or $s(j)=\gamma_{2}$, then we let $tab_G[s]=\emptyset$. In this case $\cA_G[s]=\emptyset$. Indeed, for every set $X\subseteq V(G)$ respecting Condition (1), the graph $\CG(X,E_0,s)$ contains a cycle, for all subsets $E_0\subseteq \{xv_0\mid x\in X\}$. For example, if $s(i)=\gamma_{2}$ and $s(j)=\gamma_{1}$, then the two $i$-vertices in $X$ are adjacent in $\CG(X,\emptyset,s)$ to $v_i^+$ and to the $j$-vertex in $X$, thus $\CG(X,\emptyset,s)$ contains a cycle of length 4.
	
	\item If $s(i)=s(j)=\gamma_{-2}$, then we let $tab_G[s]=\emptyset$. Similarly to Case (b), we have $\cA_G[s]=\emptyset$ because every vertex set with two $i$-vertices and two $j$-vertices induce a cycle of length 4 in $G$. 
	
	\item Otherwise, we let $tab_G[s]:=\rmc(\cA))$ with 
	\[ \cA:=  \proj(s^{-1}(\gamma_{-2})\cap\{i,j\},\acjoin(tab_H[s_H],\{ ( \{\{i,j\}\},0)\}) \] 
	where $s_H(\ell):=s(\ell)$, for $\ell \in [k]\setminus
	\{i,j\}$, and
	\begin{align*}
	(s_H(i),s_H(j)) & := \begin{cases} (\gamma_1,\gamma_1) & \textrm{if $s(i)=s(j)=\gamma_1$},\\
	 (\gamma_2,\gamma_1) & \textrm{if $(s(i),s(j)) = (\gamma_{-2},\gamma_1)$},\\
	(\gamma_1,\gamma_2)  & \textrm{if        $(s(i),s(j))=(\gamma_1,\gamma_{-2})$}.  
	\end{cases}
	\end{align*}
	Observe that this case corresponds to $s(i),s(j)\in \{\gamma_{1}, \gamma_{-2}\}$ with $s(i)=\gamma_{1}$ or $s(j)=\gamma_{1}$.
	Intuitively, we consider the weighted partitions $(p,w)\in tab_H[s_H]$ such that $i$ and $j$ belong to different blocks of $p$, we merge the blocks containing $i$ and $j$, remove the elements in $s^{-1}(\gamma_{-2})\cap\{i,j\}$ from the resulting block, and add the resulting weighted partition to $\cA$. 
	Notice that it is not necessarily to call the operator $\acreduce$ in this case since we update $tab_G[s]$ with only one table from $tab_H$. 
\end{enumerate}

\begin{lemma}\label{lem:add} Let $G=add_{i,j}(H)$ be a $k$-labeled graph. For each function $s:[k]\to \{\gamma_0,\gamma_1,\gamma_{2},\gamma_{-2}\}$, $tab_G[s]$ ac-represents
	$\cA_G[s]$ assuming that $tab_H[s']$ ac-represents $\cA_H[s']$ for all $s':[k]\to \{\gamma_0,\gamma_1,\gamma_{2},\gamma_{-2}\}$.
\end{lemma}

\begin{proof} We first recall that $lab_G(x)=lab_H(x)$ for all $x\in V(G)=V(H)$. 
	If we are in Cases (b)-(c), then we are done since we clearly have $\cA_G[s]=\emptyset$.
	Since the used operators preserve ac-representation, it is enough to prove that in Case (a), we have $\cA_G[s]=\cA_H[s]$, and in Case (d), we have $\cA_G[s]=\cA$ if we let $tab_H[s_H]= \cA_H[s_H]$.
	If we are in Case (a), then  we are done because one easily checks that $(F,E_0,(p,w))$ is a solution in $\cA_H[s]$ if and only if  $(F,E_0,(p,w))$ is a solution in $\cA_G[s]$, \ie, we have $\cA_H[s]=\cA_G[s]$.
	
	Now, we assume that we are in Case (d), that is $s(i),s(j)\in \{\gamma_{1}, \gamma_{-2}\}$ with $s(i)=\gamma_{1}$ or $s(j)=\gamma_{1}$. 
	We can assume w.l.o.g. that $s(j)=\gamma_1$.
	Let $(F,E_0,(p,w))$ be a solution in $\cA_G[s]$.
	We prove that $(p,w)\in \cA$.
	Let $\{x_j\}:=V(F)\cap lab_G^{-1}(j)$ (by assumption $s(j)=\gamma_{1}$), $X_i:=V(F)\cap lab_G^{-1}(i)$, $E_{i,j}:=\{ v x_j\mid v\in X_i\}$, and $F_H:=(V(F), E(F)\setminus E_{i,j})$. 
	Because we consider only irredundant $k$-expressions, we know that $E_{i,j}\cap E(H)=\emptyset$, \ie, $F_H$ is
	an induced forest of $H$. 
	Let $p'$ be the partition on $s_H^{-1}(\{\gamma_1,\gamma_2\})\cup\{v_0\}$ such that $(F_H,E_0,(p',w))$ is a candidate solution in
	$\cA_H[s_H]$. 
	We claim that $(F_H,E_0,s_H)$ is a solution in $\cA_H[s_H]$. 
	By the definition of $s_H$, Condition (1) is trivially satisfied. 
	If $|X_i|=1$, then 	$\CG(F_H,E_0,s_H)$ is a subgraph of $\CG(F,E_0,s)$ and so Condition (2) is satisfied. 
	And if $|X_i|\geq 2$ and $\CG(F_H,E_0,s_H)$ contains a cycle, then the cycle should contain the vertex $v_i^+\in V_{s_H}^+$, but this vertex may be replaced by $x_j$ in $\CG(F,E_0,s)$, contradicting the fact this latter is acyclic.
	Condition (3) is also satisfied because $s_H(i),s_H(j)\in \{\gamma_1,\gamma_2\}$ and for every connected component $C$ of $\CG(F_H,E_0,s_H)$, either $C$ is a connected component of $\CG(F,E_0,s)$ or $C$ contains an $\ell$-vertex with $\ell\in\{i,j\}$.
	Therefore,	$(F_H,E_0,(p',w))$ is a solution in $\cA_H[s_H]$. 
	
	It remains to prove that $(p,w)$ is added in $\cA$. 
	We claim that $\acy(p', \{\{i,j\}\}_{\uparrow V })$ holds with $V=s^{-1}(\gamma_1,\gamma_{2})\cup \{v_0\}$. 
	By definition of $\acy$ it is equivalent to prove that $i$ and $j$ cannot belong to a same block of $p'$. Assume towards a contradiction that $i$ and $j$ belong to a same block of $p'$.
	Then, there is a path, in $\CG(F_H,E_0,v_0)$, between an $i$-vertex $x_i$ and the $j$-vertex $x_j$ of $F$. 
	Let us choose this path $P$ to be the smallest one. 
	One	first notices that $P$ cannot contain the vertex $v_i^+$ of $V_{s_H}^+$, if any, because $v_i^+$ is only adjacent to $i$-vertices. 
	Because	$V(\CG(F_H,E_0,s_H))\setminus \{v_i^+\}=V(\CG(F,E_0,s))$, we would conclude that $\CG(F,E_0,s)$ contains a cycle as $x_ix_j\in E(F)\setminus E(F_H)$, contradicting that	$(F,E_0,(p,w))$ is a solution in $\cA_G[s]$. 
	Therefore, $\acy(p', \{\{i,j\}\}_{\uparrow V})$ holds. 
	By assumption $s(j)=\gamma_{1}$, thus $s^{-1}(\gamma_{-2})\cap\{i,j\} \subseteq \{i\}$.
	Since $i$ and $j$ are in the same block of the partition $p\sqcup \{ \{ i,j\}\}_{\uparrow V}$, we conclude that $(p,w)\in \proj(s^{-1}(\gamma_{-2})\cap\{i,j\},\acjoin(\{(p',w)\}, \{(\{\{i,j\}\},0)\}))$.
	\medskip
	
	It remains to prove that each  weighted partition  $(p,w)\in\cA$ belongs to $\cA_G[s]$.  Let $(F_H,E^{H}_0,(p',w))$ be a solution in $\cA_H[s_H]$ so that \[ (p,w)\in \proj(s^{-1}(\gamma_{-2})\cap\{i,j\},\acjoin(\{(p',w)\}, \{ ( \{\{i,j\}\},0)\})). \]  
	Let  $F=G[V(F_H)]$. By assumption $s(j)=\gamma_{1}$, and thus $s_H(j)=\gamma_{1}$. Let $\{x_j\}:=V(F_H)\cap lab_H^{-1}(j)$, $X_i:=V(F_H)\cap lab_H^{-1}(i)$, and $E_{i,j}:=\{ x_j v \mid v\in X_i\}$. Notice that $E(F)\setminus E(F_H)= E_{i,j}$.
	We claim that $(F,E^{H}_0, (p,w))$ is a solution in  $\cA_G[s]$.
	\begin{itemize}
		\item  First, Condition (1) is trivially satisfied by the definition of $s_H$.
		\item Secondly, $\CG(F,E^{H}_0,s)$ is a forest. Indeed, $\{i,j\}$ cannot be a block of $p'$, otherwise the $\acjoin$ operator would discard $p'$. If $|X_i|=1$,
		then $\CG(F,E^{H}_0,s)=(V(\CG(F_H,E^{H}_0,s_H)),E(\CG(F_H,E^{H}_0,s_H))\cup E_{i,j})$ and it is clearly a forest. 
		Otherwise, if $|X_i|\geq 2$, then $\CG(F,E^{H}_0,s)$ can be obtained from $\CG(F_H,E^{H}_0,s_H)$ by
		fusing the vertex $v_i^+$ and the vertex $x_j$. Clearly, this operation keeps the graph acyclic since $x_j$ and $v_i^+$ are not connected in
		$\CG(F_H,E^{H}_0,s_H)$. Thus $(F,E^{H}_0,(p,w))$ satisfies Condition (2).
		\item Each connected component of $\CG(F_H,E^{H}_0,s_H)$ is contained in a connected component of $\CG(F,E^{H}_0,s)$, and the
		$i$-vertices are in the same connected component, in $\CG(F,E^{H}_0,s)$, as $x_j$. Therefore, Condition (3) is satisfied by $(F,E^{H}_0,(p,w))$
		as $s(j)=\gamma_1$ and $s(\ell)=s_H(\ell)$ for all $\ell \in [k]\setminus \{i,j\}$.
		
		\item Also, Condition (4) is satisfied as $p$ is then obtained from $p'$ by merging the blocks of $p'$ which contains $i$ and $j$, and by removing $i$ if
		$s(i)=\gamma_{-2}$.
	\end{itemize}
	We can therefore conclude that $(F,E^{H}_0,(p,w))$ is a solution in $\cA_G[s]$.
\end{proof}

\medskip

\paragraph{\bf Computing $tab_G$ for $G=ren_{i\to j}(H)$} We can suppose that $H$ is $k$-labeled. Let $s:[k]\setminus \{i\} \to
\{\gamma_0,\gamma_1,\gamma_2,\gamma_{-2}\}$.
\begin{enumerate}[(a)]
\item Let $\cA_1:=tab_H[s_1]$ where $s_1(i)=\gamma_0$ and $s_1(\ell)=s(\ell)$ for all $\ell\in [k]\setminus \{i\}$. This set contains all weighted partitions corresponding to solutions not intersecting $lab_H^{-1}(i)$. They are trivially solutions in $\cA_G[s]$. 
\item If $s(j)= \gamma_{0}$, then let $\cA_2=\emptyset$, otherwise let $s_2:[k]\to \{\gamma_0,\gamma_1,\gamma_2,\gamma_{-2}\}$ such that $s_2(j)=\gamma_0$, $s_2(i)=s(j)$ and $s_2(\ell)=s(\ell)$ for all $\ell\in [k]\setminus \{i,j\}$ and let
  \begin{align*}
    \cA_2&:=\begin{cases} tab_H[s_2] & \textrm{if $s(j)=\gamma_{-2}$},\\  \proj(\{i\},\acjoin(tab_H[s_2], \{ ( \{\{i,j\}\},0)\})) & \textrm{otherwise.}
    \end{cases}
  \end{align*}
  This set contains all weighted partitions corresponding to solutions not intersecting $lab_H^{-1}(j)$. 
  They are solutions in $\cA_G[s]$ by replacing $i$ by $j$ with the $\acjoin$ operator, if necessary, in the corresponding weighted partitions.
  Notice that if $s(j)=\gamma_0$, then $s_1=s_2$, this is why we let $\cA_2=\emptyset$ in this case.
 
 \item If $s(j)\neq \gamma_{-2}$, then let $\cA_3:= \emptyset$, otherwise let 
 \[ \cA_3:=\bigcup_{s_3\in \cS_3}\proj(\{i,j\},tab_H[s_3]), \] 
 where $\cS_3$ is the set of functions $s_3$ with $s_3(i),s_3(j)\in \{\gamma_1, \gamma_{-2}\}$, and $s_3(\ell)=s(\ell)$ for all $\ell\in [k]\setminus \{i,j\}$. Intuitively, $\cS_3$ is the set of functions coherent with $s$ if $s(j)=\gamma_{-2}$.
 The set $\cA_3$ corresponds to partial solutions intersecting $lab_H^{-1}(i)$ and $lab_H^{-1}(j)$ when $s(j)=\gamma_{-2}$. 
 In this case, we have to ensure that the partial solutions in $\cA_3$ respect Condition (3) of Definition \ref{defn:tabfvs}. We do that by removing all the partitions with a block included in $\{i,j\}$.
 
\item We now define the last set considering the other cases. If $s(j)\neq \gamma_{2}$, then let $\cA_4=\emptyset$, otherwise let
  \begin{align*}
    \cA_4&:=\bigcup_{s_4\in \cS_4}\proj(\{i\},\acjoin(tab_H[s_4], \{(\{\{i,j\}\},0)\})),
  \end{align*}
  where $\cS_4$ is the set of functions $s_4$ with $s_4(i),s_4(j)\in \{\gamma_1, \gamma_{2}\}$ and  $s_4(\ell)=s(\ell)$ for all $\ell\in [k]\setminus \{i,j\}$. Informally, $\cS_4$ is the set of functions compatible with $s$ if $s(j)=\gamma_{2}$. 
  The set $\cA_4$ corresponds to partial solutions intersecting $lab_H^{-1}(i)$ and $lab_H^{-1}(j)$ when $s(j)=\gamma_{2}$. 
  We have to force that $i$-vertices and $j$-vertices belong to different connected components. 
  We check this with the function $\acy$ in the operator $\acjoin$.
\end{enumerate}
\medskip

We let $tab_G[s]:= \acreduce(\rmc( \cA_1 \cup \cA_2 \cup \cA_3\cup \cA_4))$.

\begin{lemma}\label{lem:ren} Let $G=ren_{i\to j}(H)$ with $H$ a $k$-labeled graph. For each $s:[k]\setminus\{i\}\to \{\gamma_0,\gamma_1,\gamma_{2},\gamma_{-2}\}$, the table $tab_G[s]$ ac-represents
  $\cA_G[s]$ assuming that $tab_H[s']$ ac-represents $\cA_H[s']$ for all $s':[k]\to \{\gamma_0,\gamma_1,\gamma_{2},\gamma_{-2}\}$. 
\end{lemma}

\begin{proof} Since the used operators preserve ac-representation, it is enough to prove that $\cA_G[s]=\cA_1\cup \cA_2\cup \cA_3\cup \cA_4$ if we let $tab_H[s']=\cA_H[s']$ for every $s': [k] \to \{\gamma_0,\gamma_1,\gamma_{2},\gamma_{-2}\}$.

  Let $(F,E_0,(p,w))$ be a solution in $\cA_G[s]$. We want to prove that $(p,w)\in\cA_1\cup \cA_2\cup \cA_3\cup \cA_4$. If $V(F)$ does not intersect $lab_H^{-1}(i)$, then $(F,E_0,(p,w))$ is a solution in $\cA_H[s_1]$. 
  Assume now that  $V(F)$ intersects $lab_H^{-1}(i)$. If $V(F)\cap lab_H^{-1}(j)=\emptyset$ and $s(j)=\gamma_{-2}$, $(F,E_0,(p,w))$ is a solution in $\cA_H[s_2]$. 
  If $V(F)\cap lab_H^{-1}(j)=\emptyset$ and $s(j)\in\{\gamma_1,\gamma_2\}$, then it is easy to check that $(F,E_0, (p',w))$ is a solution in $\cA_H[s_2]$ where $p'$ is obtained from $p$ by replacing $j$ by $i$.

  We may assume now that $V(F)\cap lab_H^{-1}(i)\ne \emptyset$, $V(F)\cap lab_H^{-1}(j)\ne \emptyset$. Then, we have $s(j)\in \{\gamma_{2},\gamma_{-2}\}$ as $|lab_G^{-1}(j)\cap V(F)|\geq 2$.  
  Let $s_\star$ be a function from $[k]$ such that $s_\star(\ell):=s(\ell)$ for all $\ell\in [k]\setminus \{i,j \}$, and for $t\in \{i,j\}$,
  \begin{align*}
  s_\star(t):=\begin{cases}
  \gamma_1 & \text{if }|V(F)\cap lab_H^{-1}(t) | = 1,\\
  s(j) & \text{if }|V(F)\cap lab_H^{-1}(t) | \geq 2.
  \end{cases}
  \end{align*}
  By definition, if $s(j)=\gamma_{-2}$, then $s_\star$ belongs to $\cS_3$ and if $s(j)=\gamma_{2}$, then $s_\star$ belongs to $\cS_4$.
  Let $p'$ be the
  partition on $s_\star^{-1}(\{\gamma_1,\gamma_2\})\cup\{v_0\}$ such that $(F,E_0,(p',w))$ is a candidate solution in $\cA_H[s_\star]$. We claim that $(F,E_0,(p',w))$ is a solution in $\cA_H[s_\star]$. By Definition of $s_\star$ and of $(F,E_0)$, Condition (1) is satisfied by $(F,E_0,(p',w))$. 
  
 	Suppose first that $s(j)=\gamma_{-2}$. 
 	Observe that Condition (2) is satisfied because the certificate graphs of $(F,E_0)$ with respect to $s$ and $s_\star$ are the same. 
  Condition (3) is also satisfied by definition of $s_\star$ and because $\CG(F,E_0,s)=\CG(F,E_0,s_\star)$.
  So, $(F,E_0,(p',w))$ is a solution in $\cA_H[s_\star]$.
 
   Assume now that $s(j)=\gamma_2$. Condition (2) is satisfied. Indeed, if $s_\star(i)=\gamma_1$, then $\CG(F,E_0,s)$ is a subgraph of $\CG(F,E_0,s_\star)$. Otherwise, if $s_\star(i)=\gamma_{2}$, then $\CG(F,E_0,s)$ can be obtained from $\CG(F,E_0,s_\star)$ by fusing $v_i^{+}$ with $v_j^+$.
   In both cases, it is easy to see that $\CG(F,E_0,s_\star)$ is acyclic as $\CG(F,E_0,s)$ is acyclic. 
   Condition (3) is satisfied because each connected component $C$ of $\CG(F,E_0,s)$ contains at least a vertex in $lab_{G}^{-1}(s^{-1}(\{\gamma_1,\gamma_2\})) \cup \{v_0\}$, and we have from the definition of $s_\star$ 
   \[ lab_{H}^{-1}(s_\star^{-1}(\{\gamma_1,\gamma_2\}))= lab_{G}^{-1}(s^{-1}(\{\gamma_1,\gamma_2\})). \] 
  
  In both cases, $(F,E_0,(p',w))$ is a solution in  $\cA_H[s_\star]$, and depending on $s(j)$, we can clearly conclude that $(p,w)$ is obtained from $(p',w)$.
\medskip

  Let us now prove that for any weighted partition $(p,w)\in\cA_1\cup \cA_2\cup \cA_3\cup \cA_4$, there is a pair $(F,E_0)$ such that $(F,E_0,(p,w))$ is a solution in $\cA_G[s]$. 
  This is clear if $(p,w)\in \cA_1\cup \cA_2$. 
  
  Assume that $s(j)=\gamma_{-2}$ and let $(p,w)\in \cA_3$. Let $s_3\in\cS_3$ and $(p',w)$ be the weighted partition from $tab_H[s_3]$ from which $(p,w)$ is obtained.
 	Let $(F,E_0,(p',w))$ be a solution in $\cA_H[s_3]$. We claim that $(F,E_0,(p,w))$ is a solution in $\cA_G[s]$. By definition of $\cS_3$, we clearly have $|V(F)\cap lab_H^{-1}(\{i,j\})|=|V(F)\cap lab_G^{-1}(j)|\geq 2$. We deduce that Condition (1) is satisfied. Condition (2) is also satisfied because $\CG(F,E_0,s_3)$ is the same as $\CG(F,E_0,s)$. 
  We claim that Condition (3) is satisfied. 
 Notice that $s^{-1}(\{\gamma_{1},\gamma_{2}\})=s_3^{-1}(\{\gamma_{1},\gamma_{2}\})\setminus\{i,j\}$. Moreover, if $ i\in s_3^{-1}(\{\gamma_{1},\gamma_{2}\})$, by definition of $s_3$, we have $s_3(i)=\gamma_1$, that is $F$ has exactly one $i$-vertex (the same statement is true for $j$).  Since we use the operator $\proj$ with $\{i,j\}$, there is no block of $p'$ included in $\{i,j\}$. 
 Hence, if $F$ contains one $i$-vertex (one $j$-vertex), then, by Condition (4), this vertex is connected in $\CG(F,E_0,s_3)$ to either $v_0$ or an $\ell$-vertex with $\ell\in s^{-1}(\{\gamma_1,\gamma_2\})$.
 We can conclude that each connected component of $F$ must contain a vertex in $lab_{G}^{-1}(s^{-1}(\{\gamma_1,\gamma_2\})) \cup \{v_0\}$, \ie Condition (3) is satisfied. Condition (4) is satisfied owing to the fact that $p'$ is obtained from $p$ by
  doing a projection on $\{i,j\}$.
  
  Assume now that $s(j)=\gamma_2$ and let $(p,w)\in \cA_4$. Let $s_4\in\cS_4$ and $(p',w)$ be the weighted partition from $tab_H[s_4]$ from which $(p,w)$ is obtained,
  and let $(F,E_0,(p',w))$ be a solution in $\cA_H[s_4]$.
  By definition of $\cS_4$, we deduce that Condition (1) is satisfied (see the case $s(j)=\gamma_{-2}$).
  If there is a cycle in $\CG(F,E_0,s)$, then it must be between an $i$-vertex and a $j$-vertex
  of $H$, but then we must have a path between them in $\CG(F,E_0,s_4)$, \ie, $i$ and $j$ belong to a same block of $p'$, contradicting that $(p,w)$ is produced
  from $(p',w)$ (because the $\acjoin$ operator would detect that $\acy(p,\{\{i,j\}\}_{\uparrow [k]})$ does not hold). 
  So, Condition (2) is satisfied. We deduce that Condition (3) is satisfied from the fact that by definition of $s_4$, we have $lab_G^{-1}(s^{-1}(\{\gamma_{1},\gamma_{2}\}))=lab_H^{-1}(s_4^{-1}(\{\gamma_{1},\gamma_{2}\}))$. Also, as $p$ is obtained from $p'$ by merging the blocks containing $i$ and $j$, and by removing $i$, we deduce that Condition (4) is satisfied. 
  
  In both cases, we can conclude that $(F,E_0,(p,w))$ is a solution in  $\cA_G[s]$.
\end{proof}

\medskip

\paragraph{\bf Computing $tab_G$ for $G=G_a\oplus G_b$} We can suppose w.l.o.g. that $G_a$ and $G_b$ are both $k$-labeled\footnote{If $J\subset [k]$ is the set of labels $G_a$ (or $G_b$), we can extend the domain of any function $s' : J\to \{\gamma_0,\gamma_{1},\gamma_{2},\gamma_{-2}\}$ to $[k]$ by setting $s'(i)=\gamma_0$ for all $i\in[k]\setminus J$.}. Let $s:[k]\to \{\gamma_0,\gamma_1,\gamma_2,\gamma_{-2}\}$. 

We say that $s_a:[k]\to\{\gamma_0,\gamma_1,\gamma_2,\gamma_{-2}\}$ and $s_b:[k]\to\{\gamma_0,\gamma_1,\gamma_2,\gamma_{-2}\}$ \emph{u-agree on $s$} if,
\begin{enumerate}[(u1)]
\item for each $i\in s_a^{-1}(\gamma_0)$, $s(i)=s_b(i)$. Similarly, for each $i\in s_b^{-1}(\gamma_0)$, $s(i)=s_a(i)$,
\item for each $i\in s^{-1}(\gamma_1)$, either $s_a(i)=\gamma_0$ or $s_b(i)=\gamma_0$,
\item for each $i\in [k]\setminus  (s_a^{-1}(\gamma_0)\cup s_b^{-1}(\gamma_0))$, if $s(i)=\gamma_2$, then $s_a(i),s_b(i)\in \{\gamma_1,\gamma_2\}$,
\item for each $i\in [k]\setminus  (s_a^{-1}(\gamma_0)\cup s_b^{-1}(\gamma_0))$, if $s(i)=\gamma_{-2}$, then $s_a(i),s_b(i)\in  \{\gamma_1,\gamma_{-2}\}$.
\end{enumerate}

The functions $s_a$ and $s_b$ inform about the indices to look at $tab_{G_a}$ and $tab_{G_b}$ in order to construct $tab_G[s]$. Let $(F,E_0,(p,w))$ be a solution in
$\cA_G[s]$, and assume that it is constructed from solutions $(F_a,E_a^0,(p_a,w_a))$ and $(F_b,E_b^0,(p_b,w_b))$, in respectively, $\cA_{G_a}[s_a]$ and $\cA_{G_b}[s_b]$. The first condition tells that if $F_a$ (resp. $F_b$) does not intersect $lab_G^{-1}(i)$, then the intersection of $F$ with $lab_G^{-1}(i)$ depends only on
$V(F_b)\cap lab_G^{-1}(i)$ (resp.  $V(F_a)\cap lab_G^{-1}(i)$), and so if $V(F)$ does not intersect $lab_G^{-1}(i)$, then $F_a$ and $F_b$ do not intersect
$lab_G^{-1}(i)$. 
The second condition tells that if $|F\cap lab_G^{-1}(i)|=1$, then either $F\cap lab_G^{-1}(i)=F_a\cap lab_G^{-1}(i)$ or $F\cap lab_G^{-1}(i)=F_b\cap lab_G^{-1}(i)$. 
The other two conditions tell when $F$ intersects both $lab_{G_a}^{-1}(i)$ and $lab_{G_b}^{-1}(i)$. 
Notice that we may have $s(i)$ set to $\gamma_{-2}$ (or $\gamma_{2}$), while $F_a$ and $F_b$ each intersects $lab_G^{-1}(i)$ in exactly one vertex.

We let $tab_G[s]:=\acreduce(\rmc(\cA))$ where, 
{\small
\begin{align*}
  \cA &:= \bigcup\limits_{\substack{s_a,s_b \\\textrm{~u-agree on~$s$}}}\acjoin\left(\proj(s^{-1}\left(\gamma_{-2}),tab_{G_a}[s_a]\right), \proj(s^{-1}\left(\gamma_{-2}),tab_{G_b}[s_b]\right)\right).
\end{align*}}
The weighted partitions $(p,w)$ added in $tab_G[s]$ are all the weighted partitions that are joins of weighted partitions $(p_a,w_a)$ and $(p_b,w_b)$ from
$tab_{G_a}[s_a]$ and $tab_{G_b}[s_b]$, respectively. We need to do the projections before the join because we may have $s_a(i)=s_b(i)=\gamma_1$,
$s(i)=\gamma_{-2}$, and there is $j$ such that $s(j)=\gamma_2$ with $j$ in the same block as $i$ in both partitions $p_a$ and $p_b$.  If we do the projection after the
$\acjoin$ operator, this latter will detect that $\acy(p_a,p_b)$ does not hold, and won't construct $(p,w)$, which indeed corresponds to a solution in $\cA_G[s]$.

\begin{lemma}\label{lem:join} Let $G=G_a\oplus G_b$ be a $k$-labeled graph. For each function $s:[k]\to \{\gamma_0,\gamma_1,\gamma_2,\gamma_{-2}\}$, the table $tab_G[s]$ac-represents
  $\cA_G[s]$ assuming that $tab_{G_a}[s']$ and $tab_{G_b}[s']$ ac-represent, respectively, $\cA_{G_a}[s']$ and $\cA_{G_b}[s']$, for each
  $s':[k]\to \{\gamma_0,\gamma_1,\gamma_2,\gamma_{-2}\}$.
\end{lemma}

\begin{proof} Since the used operators preserve ac-representation, it is enough to prove that $\cA_G[s]=\cA$ if we let $tab_{G_t}[s']=\cA_{G_t}[s']$, for every $t\in\{a,b\}$ and $s' : [k]\to \{\gamma_0,\gamma_1,\gamma_{2},\gamma_{-2}\}$.

  Let $(F,E_0,(p,w))$ be a solution in $\cA_G[s]$. We claim that $(p,w)\in\cA$. For $t\in \{a,b\}$, let $F_t:=G_t[V(F)\cap V(G_t)]$, $E^{t}_0:=\{v_0v \in E_0 \mid v\in V(F_t)\}$, and
  $w_t:=\wc(V(F_t))$, and let $s_t:[k]\to \{\gamma_0,\gamma_1,\gamma_2,\gamma_{-2}\}$ such that
  \begin{align*}
    s_t(i)&:=\begin{cases} \gamma_0 & \textrm{if $V(F_t)\cap lab_{G_t}^{-1}(i)=\emptyset$},\\
     \gamma_1 & \textrm{if $|V(F_t)\cap lab_{G_t}^{-1}(i)|=1$},\\
      s(i) & \textrm{if $|V(F_t)\cap lab_{G_t}^{-1}(i)|\geq 2$.} \end{cases} 
  \end{align*}

  It is straightforward to verify that $s_a$ and $s_b$ u-agree on $s$. 
  Observe that, by definition, $s_t^{-1}(\gamma_{-2})\subseteq s^{-1}(\gamma_{-2})$ and $s_t^{-1}(\gamma_{2})\subseteq s^{-1}(\gamma_{2})$, for each $t\in\{a,b\}$.
  Let $p_a$ and $p_b$ be, respectively, partitions on $s_a^{-1}(\{\gamma_1,\gamma_2\})\cup \{v_0\}$
  and $s_b^{-1}(\{\gamma_1,\gamma_2\})\cup \{v_0\}$ such that $(F_a,E^{a}_0,(p_a,w_a))$ and $(F_b,E^{b}_0,(p_b,w_b))$ are, respectively, candidate solutions in
  $\cA_{G_a}[s_a]$ and $\cA_{G_b}[s_b]$.
  
  We claim that $(F_a,E^{a}_0,(p_a,w_a))$ is a solution in $\cA_{G_a}[s_a]$. By definition of $p_a$, $s_a$, and of $(F_a,E^{a}_0)$, Conditions (1) and (4) are clearly  satisfied. Because $s_a^{-1}(\gamma_2) \subseteq s^{-1}(\gamma_2)$, and $F=F_a\oplus F_b$, we can conclude that $\CG(F_a,E^{a}_0,s_a)$ is an induced subgraph of
  $\CG(F,E_0,s)$, and because $\CG(F,E_0,s)$ is acyclic, we can conclude that $\CG(F_a,E^{a}_0,s_a)$ is acyclic, \ie, Condition (2) is satisfied.  If a connected component $C$ of $\CG(F_a,E^{a}_0,s_a)$ does not intersect  $s_a^{-1}(\{\gamma_1,\gamma_2\})\cup \{v_0\}$, then $C$ is entirely contained in $lab_{G_a}^{-1}(s_a^{-1}(\gamma_{-2}))$. But, this yields a contradiction with $(F,E_0,(p,w))$ satisfying Condition (3) because $C$ is a connected component of $\CG(F,E_0,s)$, and $s_a^{-1}(\gamma_{-2})\subseteq s^{-1}(\gamma_{-2})$. Therefore Condition (3) is also satisfied. We can thus conclude that $(F_a,E^{a}_0,(p_a,w_a))$ is a solution in $\cA_{G_a}[s_a]$. Similarly, one can check that
  $(F_b,E^{b}_0,(p_b,w_b))$ is a solution in $\cA_{G_b}[s_b]$.

  It remains to prove that {\small
    \[ (p,w)\in\acjoin(\proj(s^{-1}(\gamma_{-2}),\{(p_a,w_a)\}),\proj(s^{-1}(\gamma_{-2}),\{(p_b,w_b)\})). \]} First, recall that each connected component of $F$ is
  either a connected component of $F_a$ or of $F_b$. Then, because $(F,E_0,(p,w))$ satisfies Condition (3), we have that
  $\proj(s^{-1}(\gamma_{-2}),\{(p_a,w_a)\}))\neq \emptyset$, and similarly $\proj(s^{-1}(\gamma_{-2}),\{(p_b,w_b)\}))\neq \emptyset$. 
  We deduce that 
  \begin{align*}
  	p'_a:=p_{a\downarrow (s^{-1}(\gamma_{-2}))}\in \proj(s^{-1}(\gamma_{-2}),\{(p_a,w_a)\}),\\
  	p'_b:=p_{b\downarrow (s^{-1}(\gamma_{-2}))}\in \proj(s^{-1}(\gamma_{-2}),\{(p_b,w_b)\}).
  \end{align*}
  
  Let
  $p''_a:= p'_{a\uparrow ([k]\setminus s^{-1}(\{ \gamma_0, \gamma_{-2}\}))}$ and $p''_b:=p'_{b\uparrow ([k]\setminus s^{-1}(\{ \gamma_0, \gamma_{-2}\}))}$.  We claim that $\acy(p''_1,p''_2)$ holds. Assume towards a contradiction that it is not the case. We let
  $\sim_a$ (resp. $\sim_b$) be an equivalence relation on $[k]\setminus s^{-1}(\{ \gamma_0, \gamma_{-2}\})\cup \{v_0\}$ where $i\sim_a j$ (resp. $i\sim_b j$) if there is an $i$-vertex\footnote{We consider $v_0$ as a $v_0$-vertex.} and a $j$-vertex that are
  connected in $\CG(F_a,E^{a}_0,s_a)$ (resp.  $\CG(F_b,E^{b}_0,s_b)$). By the graphical definition of $\acy$, we can easily see that if $\acy(p''_a,p''_b)$ does not
  hold, then there is a sequence $i_0,\dots, i_{2r-1}$ of $s^{-1}(\{\gamma_1,\gamma_2\})\cup \{v_0\}$ such that\footnote{ The indexes are modulo $2r$.},
  {$\text{for all } 0 \leq \alpha < r-1, \text{ we have } i_{2\alpha}\sim_b i_{2\alpha+1} \text{ and } i_{2\alpha+1}\sim_a
      i_{2\alpha+2}.$}
  We can thus construct a cycle in $\CG(F,E_0,s)$ from this sequence since $V(F_a)\cap V(F_b) = \emptyset$, $\CG(F_a,E^{a}_0,s_a)$ and $\CG(F_b,E^{b}_0,s_b)$ are
  induced subgraphs of $\CG(F,E_0,s)$, and all the vertices labeled with a label from $s^{-1}(\gamma_2)$ are at distance at most two in $\CG(F,E_0,s)$. This yields a
  contradiction as $\CG(F,E_0,s)$ is acyclic by assumption. Therefore, $\acy(p''_1,p''_2)$ holds.
  
  Finally, $p=p''_a\sqcup p''_b$ because one easily checks that there is an $i$-vertex $x$ connected to a $j$-vertex $y$ in $\CG(F,E_0,s)$ if and only if $i\cR j$ where $\cR$ is
  the transitive closure of ($i\sim_a j$ or $i\sim_b j$). This follows from the fact that for every $i\in [k]\setminus s^{-1}(\{ \gamma_0, \gamma_{-2}\})$, either there is exactly one $i$-vertex in $F$ or the $i$-vertices of $F$ are all adjacent to $v_i^+$ in $\CG(F,E_0,s)$. In both cases, the $i$-vertices of $F$ are in the same connected component of  $\CG(F,E_0,s)$. Since, the equivalence classes of $\cR$ correspond to the blocks of $p''_1\sqcup p''_2$, and $w=w_a+w_b$, we can conclude that  $(p,w)\in\acjoin(\proj(s^{-1}(\gamma_{-2}),\{(p_a,w_a)\}),\proj(s^{-1}(\gamma_{-2}),\{(p_b,w_b)\}))$.
  
  \medskip

  We now prove that if $(p,w)$ is added to $tab_G[s]$ from $(p_a,w_a)\in\cA_{G_a}[s_a]$ and $(p_b,w_b)\in\cA_{G_b}[s_b]$, then there exists a  pair $(F,E_0)$ such that $(F,E_0,(p,w))$ is a solution in $\cA_G[s]$.  Let $(F_a,E^{a}_0)$ and $(F_b,E^{b}_0)$ such that $(F_a,E^{a}_0,(p_a,w_a))$ and
  $(F_b,E^{b}_0,(p_b,w_b))$ are solutions in, respectively, $\cA_{G_a}[s_a]$ and $\cA_{G_b}[s_b]$ with $s_a$ and $s_b$ u-agreeing on $s$.  We claim that
  $(F,E_0,(p,w))$ is a solution in $\cA_G[s]$ with $F:=(V(F_1)\cup V(F_2),E(F_1)\cup E(F_2))$ and $E_0:=E^{a}_0\cup E^{b}_0$.  Because $s_a$ and $s_b$ u-agree on
  $s$, we clearly have that Condition (1) is satisfied.

  Let $\sim_a$ and $\sim_b$ as defined above.  Assume towards a contradiction that there exists a cycle $C$ in $\CG(F,E_0,s)$. Since both $\CG(F_a,E^{a}_0,s_a)$ and
  $\CG(F_b,E^{b}_0,s_b)$ are acyclic, $C$ must be a cycle alternating between paths in $\CG(F_a,E^{a}_0,s_a)$ and paths in $\CG(F_b,E^{b}_0,s_b)$. One can easily
  check that this implies the existence of a sequence $i_0,\dots, i_{2r-1}$ of $s^{-1}(\{\gamma_1,\gamma_2\})\cup \{v_0\}$ such that\footnotemark[\value{footnote}],
   {$\text{for all } 0 \leq \alpha < r-1, \text{ we have } i_{2\alpha}\sim_b i_{2\alpha+1} \text{ and } i_{2\alpha+1}\sim_a i_{2\alpha+2}.$}
  Moreover, it is easy to infer, from this sequence and the graphical definition of $\acy$, that
  $\acy\left(p'_{a\uparrow V},p'_{b\uparrow V}\right)$ does not hold with $p'_a:=p_{a\downarrow (s^{-1}(\gamma_{-2}))}$, $p'_b:=p_{b\downarrow (s^{-1}(\gamma_{-2}))}$ and $V=([k]\setminus s^{-1}(\{\gamma_0,\gamma_{-2}\}))$, contradicting the fact that
  $(p,w)=\acjoin(\proj(s^{-1}(\gamma_{-2}),\{(p_a,w_a)\}), \proj(s^{-1}(\gamma_{-2}),\{(p_b,w_b)\}))$. Therefore, $\CG(F,E_0,s)$ is acyclic and so Condition (2) is satisfied.
  
  
  If we suppose that Condition (3) is not satisfied, then there is a connected component $C$ of $\CG(F,E_0,s)$ that does not intersect $s^{-1}(\{\gamma_1,\gamma_2\})\cup\{v_0\}$,
  \ie, $C$ is fully contained in $lab_G^{-1}(s^{-1}(\gamma_{-2}))$. Since $F=F_a\oplus F_b$, $C$ is either a connected component of $F_a$ or of $F_b$. Suppose
  \emph{w.l.o.g.} that $C$ is a connected component of $F_a$. Observe that $C$ intersects $lab_{G_a}^{-1}(s_a^{-1}(\{\gamma_1,\gamma_2\}))$ because
  $(F_a,E^{a}_0,(p_a,w_a))$ is a solution in $\cA_{G_a}[s_a]$.  Moreover, $C$ does not intersect $lab_{G_a}^{-1}(s_a^{-1}(\gamma_2))$, otherwise $C$ would intersect
  $lab_{G}^{-1}(s^{-1}(\gamma_2))$ since if $s_a(i)=\gamma_2$, then $s(i)=\gamma_2$, for all $i\in[k]$. Thus $C$ is a connected component of $\CG(F_a,E^{a}_0,s_a)$ and
  $b_C:=\{i\in s_a^{-1}(\gamma_1)\mid C\cap lab_{G_a}^{-1}(i)\neq \emptyset\}$ is a block of $p_a$ because $(F_a,E^{a}_0,(p_a,w_a))$ is a candidate solution in $\cA_{G_a}[s_a]$. Thus, by
  definition of $\proj$, we have $\proj(s^{-1}(\gamma_{-2}),\{(p_a,w_a)\})=\emptyset$, which contradicts the fact that
  $(p,w)=\acjoin(\proj(s^{-1}(\gamma_{-2}),\{(p_a,w_a)\}), \proj(s^{-1}(\gamma_{-2}),\{(p_b,w_b)\}))$. So, Condition (3) is also satisfied.
  We deduce that Condition (4) is satisfied by observing that, for $i,j\in s^{-1}(\{\gamma_{1},\gamma_{2}\})\cup \{v_0\}$, there is an $i$-vertex connected to a $j$-vertex in $\CG(F,E_0,s)$ if and only if $i\cR j$ where $\cR$ is the transitive closure of ($i\sim_a j$ or $i\sim_b j$).
  This concludes the proof that  $(F,E_0,(p,w))$ is a solution in $\cA_G[s]$. 
\end{proof}

\begin{theorem}\label{thm:fvs} There is an algorithm that, given an $n$-vertex graph $G$ and an irredundant $k$-expression of $G$, computes a  minimum feedback vertex set in time $O(15^k\cdot 2^{(\omega + 1)\cdot k}\cdot k^{O(1)}\cdot n)$.

\end{theorem}

\begin{proof} We do a bottom-up traversal of the  $k$-expression and at each step we update the tables as indicated above. The correctness of the algorithm
  follows from Lemmas \ref{lem:add}-\ref{lem:join}. 
  From the definition of $\cA_G[s]$, we conclude that the size of a maximum induced forest is the maximum over all $s:[k]\to \{\gamma_0,\gamma_1,\gamma_2,\gamma_{-2}\}$ with $s^{-1}(\gamma_{2})=\emptyset$, of $\max\{w \mid (\{\{v_0\}\cup s^{-1}(\gamma_{1}) \} \},w)\in tab_G[s]\}$ because $tab_G[s]$ ac-represents $\cA_G[s]$ for all $s:[k]\to \{\gamma_0,\gamma_1,\gamma_2,\gamma_{-2}\}$.
  
  Let us discuss the time complexity now.  If $G=add_{i,j}(H)$ or $G=ren_{i\to j}(H)$, and
  $s:[k]\to \{\gamma_0,\gamma_1,\gamma_2,\gamma_{-2}\}$, then we update $tab_G[s]$ from a constant number of tables from $tab_H$, each identified in constant time from
  $s$. Since each table contains at most $ 2^{k-1}\cdot k$ entries, we call the function $\acreduce$ with a set of size at most $O(2^{k-1}\cdot k)$ as input. By Theorem \ref{thm:reduce1}, we can thus update $tab_G$ in time $2^{\omega\cdot k}\cdot k^{O(1)}$.  
  If $G=G_a\oplus G_b$, then we claim that the tables from $tab_G$ are computable in time $O(15^k\cdot 2^{(\omega + 1)\cdot k}\cdot k^{O(1)})$.
  For $s:[k]\to \{\gamma_0,\gamma_1,\gamma_2,\gamma_{-2}\}$, we let
  {\small
  	\begin{align*}
  	\cA[s] &:= \bigcup\limits_{\substack{s_a,s_b \\\textrm{~u-agree on~$s$}}}\acjoin\left(\proj(s^{-1}\left(\gamma_{-2}),tab_{G_a}[s_a]\right), \proj(s^{-1}\left(\gamma_{-2}),tab_{G_b}[s_b]\right)\right).
  	\end{align*}}
  By Theorem \ref{thm:reduce1}, computing $tab_G[s]:=\acreduce(\rmc(\cA[s]))$ can be done in time $|\cA[s]|\cdot 2^{(\omega-1)\cdot k}\cdot k^{O(1)}$.
  Therefore, we can compute the tables from $tab_G$ in time
  \begin{align*}
  	\sum_{s:[k]\to  \{\gamma_0,\gamma_1,\gamma_2,\gamma_{-2}\}} |\cA[s] |\cdot 2^{(\omega-1)\cdot k}\cdot k^{O(1)}.
  \end{align*}
 Now, observe that there are at most $15^k$ functions $s,s_a,s_b : [k]\to \{\gamma_0,\gamma_1,\gamma_2,\gamma_{-2}\}$ such that $s_a$ and $s_b$ u-agree on $s$. 
 Indeed, for all $i\in[k]$, if $s_a$ and $s_b$ u-agree on $s$, then the tuple $(s_a(i),s_b(i),s(i))$ can take up to 15 values. See Table \ref{tab:s_as_bs} for all the possible values.
 \renewcommand{\arraystretch}{1.5}
 \begin{table}[h]
 	\centering
 	\begin{tabular}{|c|c|c|c|c|}
 		\hline
 		& $s_b(i)=\gamma_{0}$ & $s_b(i)=\gamma_{1}$ & $s_b(i)=\gamma_{2}$& $s_b(i)=\gamma_{-2}$\\
 		\hline
 		$s_a(i)=\gamma_{0}$ & $\gamma_{0}$ & $\gamma_{1}$ & $\gamma_{2}$ & $\gamma_{-2}$\\
 		\hline
 		$s_a(i)=\gamma_{1}$ & $\gamma_{1}$ & $\gamma_{2},\gamma_{-2}$ & $\gamma_{2}$ & $\gamma_{-2}$\\
 		\hline
 		$s_a(i)=\gamma_{2}$ & $\gamma_{2}$ & $\gamma_{2}$ & $\gamma_{2}$ & forbidden\\
 		\hline
 		$s_a(i)=\gamma_{-2}$ & $\gamma_{-2}$ & $\gamma_{-2}$ & forbidden & $\gamma_{-2}$\\
 		\hline
 	\end{tabular}
 	\caption{Possibles values of $s(i)$ depending on the value of $s_a(i)$ and $s_b(i)$ when $s_a$ and $s_b$ u-agree on $s$, there are 15 possible values for the tuple $(s_a(i),s_b(i),s(i))$.}
 	\label{tab:s_as_bs}
 \end{table}

 Because each table of $tab_{G_a}$ and $tab_{G_b}$ contains at most $2^{k-1}\cdot k$ values, we have $|\acjoin\left(\proj(s^{-1}\left(\gamma_{-2}),tab_{G_a}[s_a]\right), \proj(s^{-1}\left(\gamma_{-2}),tab_{G_b}[s_b]\right)\right)| \leq  2^{2k-2}\cdot k^2$.
 It follows that $\sum_{s:[k]\to  \{\gamma_0,\gamma_1,\gamma_2,\gamma_{-2}\}} |\cA[s] | \leq 15^k \cdot 2^{2k}\cdot k^2.$
 	Hence, we can conclude that the tables from $tab_G$ can be computed in time $O(15^k\cdot 2^{(\omega + 1)\cdot k}\cdot k^{O(1)})$.
 	
	Because the size of a $k$-expression is $O(n\cdot k^2)$, we can conclude that a minimum feedback vertex set can be computed in the given time.
\end{proof}

\section{Connected (Co-)$(\sigma,\rho)$-Dominating Sets}\label{sec:dom}

We will show here how to use the operators defined in \cite{BodlaenderCKN15} in order to obtain a $2^{O(d\cdot k)}\cdot n$ time algorithm for computing a minimum or a maximum connected $(\sigma,\rho)$-dominating set, given a $k$-expression, with $d$ a constant that depends only on $(\sigma,\rho)$.
We deduce from this algorithm a $2^{O(k)}\cdot n^{O(1)}$ time algorithm for computing a minimum node-weighted Steiner tree, and a $2^{O(d\cdot k)}\cdot n^{O(1)}$ time algorithm for computing a maximum (or minimum) connected co-$(\sigma,\rho)$-dominating set. 

We let $opt\in \{\min,\max\}$, \ie, we are interested in computing a connected $(\sigma,\rho)$-dominating set of maximum (or minimum) weight if $opt=\max$ (or $opt=\min$). Let us first give some definitions.
As defined in Section \ref{sec:framework}, $\rmc$ works only for the case $opt=\max$, we redefine it as follows in order to take into account minimization problems.
$$ \rmc(\cA) :=\{(p,w)\in \cA \mid \forall (p,w')\in \cA, opt(w,w')=w\}.$$ 

\paragraph{\bf Join.} Let $V'$ be a finite set. For $\cA\subseteq \pis(V)\times \bN$ and $\cB\subseteq \pis(V')\times \bN$, we define $\join(\cA,\cB)\subseteq
\pis(V\cup V')\times \bN$ as {
\begin{align*}
  \join(\cA,\cB)&:=\{(p_{\uparrow V'}\sqcup q_{\uparrow V},w_1+w_2)\mid (p,w_1)\in \cA, (q,w_2)\in \cB\}.
\end{align*}}
This operator is the one from \cite{BodlaenderCKN15}. It is used mainly to construct partial solutions of $G\oplus H$ from partial solutions of $G$ and $H$.

The following proposition assumes that $\log(|\cA|) \leq |V|^{O(1)}$ for each $\cA\subseteq \pis(V)\times \bN$ (this can be established by applying the operator $\rmc$).

\begin{proposition}[Folklore]\label{prop:join} The operator $\join$ can be performed in time $|\cA|\cdot |\cB|\cdot |V\cup V'|^{O(1)}$ and the size of its output is upper-bounded by $|\cA|\cdot |\cB|$. 
\end{proposition}

The following is the same as Definition \ref{defn:par-rep}, but does not require acyclicity. 

\begin{definition}[\cite{BodlaenderCKN15}]\label{defn:par-rep-j} For $\cA\subseteq \pis(V)\times \bN$, with $V$ a finite set, and $q\in \pis(V)$, let 
  {$$ \opt(\cA,q) := opt \{w\mid (p,w)\in \cA,p\sqcup q = \{V\} \}.$$}
  A set of weighted partitions $\cB\subseteq \pis(V)\times \bN$ \emph{represents} $\cA$ if for each $q\in \pis(V)$, it holds that $\opt(\cA,q)=\opt(\cB,q)$.

  Let $Z$ and $V'$ be two finite sets.  A function $f: 2^{\pis(V)\times \bN}\times Z\to 2^{\pis(V')\times \bN}$ is said to \emph{preserve representation} if for each
  $\cA,\cB\subseteq \pis(V)\times \bN$ and $z\in Z$, it holds that $f(\cB,z)$ represents $f(\cA,z)$ whenever $\cB$ represents $\cA$.
\end{definition}

\begin{lemma}[\cite{BodlaenderCKN15}]\label{lem:par-repj} The operators $\rmc$, $\proj$ and $\join$ preserve representation.
\end{lemma}

\begin{theorem}[\cite{BodlaenderCKN15}]\label{thm:reduce2} There exists an algorithm $\reduce$ that, given a set of weighted partitions
  $\cA\subseteq\Pi(V)\times \mathbb{N}$, outputs in time $|\cA|\cdot 2^{(\omega-1)|V|}\cdot |V|^{O(1)}$ a subset $\cB$ of $\cA$ that represents
  $\cA$, and such that $|\cB|\leq 2^{|V|-1}$.
\end{theorem}

We use the following function to upper bound the amount of information we need to store in our dynamic programming tables concerning the $(\sigma,\rho)$-domination.
For every non-empty finite or co-finite subset $\mu\subseteq \bN$, we define $d(\mu)$ such as 
\begin{align*}
d(\mu):=\begin{cases}
0 & \text{if } \mu=\bN,\\
1+ \min (\max(\mu), \max(\bN\setminus \mu))  &\text{otherwise.}
\end{cases}
\end{align*}
For example, $d(\bN^+)=1$ and for every $c\in \bN$, we have $d( \{0,\dots,c\})=c+1$.

 The definition of $d$ is motivated by the following observation which is due to the fact that, for all $\mu\subseteq \bN$, if $d(\mu)\in \mu$, then $\mu$ is co-finite and contains $\bN\setminus\{1,\dots,d(\mu)\}$.
\begin{fact}\label{fact:d}
	For every $a,b\in \bN$ and $\mu$ a finite or co-finite subset of $\bN$, we have $a+b\in \mu$ if and only if $\min(d, a+b)\in\mu$.
\end{fact}

Let us describe with a concrete example the information we need concerning the $(\sigma,\rho)$-domination.
We say that a set $D\subseteq V(G)$ is a $2$-dominating set if every vertex in $V(G)$ has at least $2$ neighbors in $D$. 
It is worth noticing that a 2-dominating set is an $(\bN\setminus\{0,1\}, \bN\setminus\{0,1\})$-dominating set and $d(\bN\setminus\{0,1\})=2$.
Let $H$ be a $k$-labeled graph used in an irredundant $k$-expression of a graph $G$.
Assuming $D_H\subseteq V(H)$ is a subset of a 2-dominating set $D$ of $G$, we would like to characterize the sets $Y\subseteq V(G)\setminus V(H)$ such that $D\cup Y$ is a $2$-dominating set of $G$. 
One first observes that $D_H$ is not necessarily a 2-dominating set of $H$, and $D\setminus D_H$ also is not necessarily a 2-dominating set of $G[V(G)\setminus V(H)]$.
Since we want to 2-dominate $V(H)$, we need to know for each vertex $x$ in $V(H)$ how many neighbors it needs in addition to be 2-dominated by $D_H$.
For doing so, we associate $D_H$ with a sequence $R'=(r_1',\dots,r_k')$ over $\{0,1,2\}^k$ such that, for each $i\in [k]$, every vertex in $lab_H^{-1}(i)$ has at least $2-r_i'$ neighbors in $D_H$. 
For example, if $r_i'=1$, then every $i$-vertex has at least one neighbor in $D_H$. 
Notice that $D_H$ can be associated with several such sequences.
This sequence is enough to characterize what we need to 2-dominate $V(H)$ since the vertices with the same label in $H$ have the same neighbors in the graph $(V(G), E(G)\setminus E(H))$ and each vertex needs at most 2 additional neighbors to be 2-dominated.

In order to update the sequence $R'=(r_1',\dots,r_k')$ associated with a set $D_H$, we associate with $D_H$ another sequence $R=(r_1,\dots,r_k)$ over $\{0,1,2\}^k$ such that $r_i$ corresponds to the minimum between $2$ and the number of $i$-vertices in $D_H$, for each $i\in[k]$.
This way, when we apply an operation $add_{i,j}$ on $H$, we know that every $i$-vertex has at least $2-r_i'+r_j$ neighbors in $D_H$ in the graph $add_{i,j}(H)$. 
For example, if $r_j=1$ and every $i$-vertex has at least one neighbor in $H$ that belongs $D_H$, then we know that every $i$-vertex is 2-dominated by $D_H$ in the graph $add_{i,j}(H)$.

\medskip

Let $(\sigma,\rho)$ be a fixed pair of non-empty finite or co-finite subsets of $\bN$. 
Let's first show how to compute an optimum connected $(\sigma,\rho)$-dominating set. We consider node-weighted Steiner tree and connected co-$(\sigma,\rho)$-dominating set at the end of the section.
Let $d:=\max \{d(\sigma),d(\rho)\}$.

The following definitions formalize the intuitions we give for 2-dominating set to the $(\sigma,\rho)$-domination.

\begin{definition}[Certificate graph of a solution]\label{def:certifdom}
	Let $G$ be a $k$-labeled graph, and $R':=(r'_1,\ldots,r'_k)\in \{0,\dots,d\}^k$.
	Let	$V^+:=\{v_1^1,\ldots, v_d^1,v_1^2,\ldots, v_d^2,\ldots, v^k_1,\ldots,v_d^k\}$ be a set disjoint from $V(G)$ and of size $d\cdot k$.
	Let $V^+(R'):=V^+_1(R') \cup \cdots \cup V^+_k(R')$ with
	\begin{align*}
	V^+_i(R') &:=\begin{cases} \emptyset & \textrm{if $r_i'=0$},\\ 
	\{v_1^i,\ldots,v_{r'_{i}}^i\} & \textrm{otherwise}. \end{cases}
	\end{align*}
	The certificate graph of $G$ with respect to $R'$, denoted by $\CG(G,R')$, is the graph $(V(G)\cup V^+(R'),E(G) \cup  E^+_1\cup \cdots \cup E^+_k)$ with \[ E^+_i=\{ \{v, v^i_t \} \mid v_t^i\in V^+_i(R') \wedge v\in lab_G^{-1}(i)\}. \] It is worth noticing that $E_i^+$ is empty if $lab_G^{-1}(i)=\emptyset$ or $V_i^+(R')=\emptyset$. 
\end{definition} 

\begin{definition}\label{def:R}
	Let $G$ be a $k$-labeled graph. For each $D\subseteq V(G)$ and $i\in [k]$, let $
	r_{i,G}^d(D):=\min(d,|lab_G^{-1}(i)\cap D|)$
	and let $r_G^d(D):=(r^d_{1,G}(D),\ldots,r^d_{k,G}(D))$. 
\end{definition}
The sequence $r^d_G(D)$ describes how each label class is intersected by $D$ up to $d$ vertices. 
Moreover, notice that $|\{r^d_{G}(D)\mid D\subseteq V(G)\}|\leq  |\{0,\dots,d\}^k|\leq (d+1)^k$.
\medskip

The motivation behind these two sequences is that for computing an optimum $(\sigma,\rho)$-dominating set, it is enough to compute, for any $k$ labeled graph $H$ used in an irredundant $k$-expression of a graph $G$ and for each $R,R'\in\{0,\dots,d\}^k$, the optimum weight of a set $D\subseteq V(H)$ such that
\begin{itemize}
	\item $r_H(D)=R$,
	\item $D\cup V^+(R')$ $(\sigma,\rho)$-dominates $V(H)$ in the graph $\CG(H,R')$.
\end{itemize}
It is worth noticing that the sequences $r_H^d(D)$ and $R'$ are similar to the notion of \emph{$d$-neighbor equivalence} introduced in \cite{BuiXuanTV13}.

\medskip
We can assume w.l.o.g. that $d\ne 0$, that is $\sigma\ne \bN$ or $\rho\ne \bN$. 
Indeed, if $\sigma=\rho=\bN$, then the problem of finding a minimum (or maximum) (co-)connected $(\sigma,\rho)$-dominating set is trivial.
For computing an optimum connected $(\sigma,\rho)$-dominating set, we will as in Section \ref{sec:fvs} keep partitions of a subset of labels corresponding to the connected components of the sets $D$ (that are candidates for the $(\sigma,\rho)$-domination).
As $d\neq 0$, we know through $r_H(D)$ the label classes intersected by $D$. 
Moreover, we know through $R'=(r_1',\dots,r_k')$ whether the $i$-vertices in such a $D$ will have a neighbor in any extension $D'$ of $D$ into a $(\sigma,\rho)$-dominating set.
It is enough to keep the partition of the labels $i$ with $r_{i,H}(D)\neq 0$ and $r_i'\neq 0$ that corresponds to the equivalence classes of the equivalence relation $\sim$ on $\{i\in[k]\setminus r_i\ne 0 \text{ and } r_i'\ne 0\}$ where $i\sim j$ if and only if an $i$-vertex is connected to a $j$-vertex in $\CG(G,R')[D\cup V^{+}]$.

We use the following definition to simplify the notations.
\begin{definition}\label{def:active}
	For $R=(r_1,\dots,r_k),R'=(r_1',\dots,r_k')\in \{0,\dots,d\}^k$, we define $\Active(R,R')=\{i\in[k] \mid r_i\ne 0 \text{ and } r_i'\ne 0\}$.
\end{definition}

We are now ready to define the sets of weighted partitions which representative sets we manipulate in our dynamic programming tables.

\begin{definition}[{Weighted partitions in $\cD_G[R,R']$}]\label{defn:tabdom}
 	Let $G$ be a $k$-labeled graph, and $R,R'\in \{0,\ldots,d\}^k$. 
	The entries of $\cD_G[R,R']$ are all the weighted partitions $(p,w)\in \Pi(\Active(R,R'))\times \bN$ so that there exists a set $D\subseteq V(G)$ such that $\wc(D)=w$, and
	\begin{enumerate}
		\item $r^d_G(D)=R$,
		\item $D\cup V^+(R')$ $(\sigma,\rho)$-dominates $V(G)$ in $\CG(G,R')$,
		\item if $\Active(R,R')=\emptyset$, then $G[D]$ is connected, otherwise for each connected component $C$ of $G[D]$, we have $C\cap lab_G^{-1}(\Active(R,R'))\neq \emptyset$,
		\item $p=\Active(R,R')/ \sim$ where $i\sim j$ if and only if an $i$-vertex is connected to a $j$-vertex in $\CG(G,R')[D\cup V^{+}(R')]$.
	\end{enumerate} 
\end{definition}
Conditions (1) and (2) guarantee that $(p,w)$ corresponds to a set $D$ that is coherent with $R$ and $R'$. 
Condition (3) guarantees that each partial solution can be extended into a connected graph. 
Contrary to Section \ref{sec:fvs}, the set of labels expected to play a role in the connectivity (\ie $\Active(R,R')$) can be empty. 
In this case, we have to make sure that the weighted partitions represent a connected solution. 
It is worth mentioning that $G[\emptyset]$, \ie the empty graph, is considered as a connected graph.
Observe that for each $(p,w)\in\cD_G[R,R']$, the partition $p$ has the same meaning as in Section \ref{sec:fvs}.
 It is worth noticing that $\sim$ is an equivalence relation because if $i\in \Active(R,R')$, then all the vertices in $lab_G^{-1}(i)$ are connected in $\CG(G,R')$ through the vertex in $V^+_i(R')$.
 In fact, the relation $\sim$ is equivalent to the transitive closure of the relation $\asymp$ where $i \asymp j$ if there exists an $i$-vertex and a $j$-vertex in the same connected component of $G[D]$. 
\medskip

In the sequel, we call a pair $(D,(p,\wc(D)))$ a \emph{candidate solution} in $\cD_G[R,R']$ if $p=\Active(R,R')/ \sim$ where $i\sim j$ if and only if an $i$-vertex is connected to a $j$-vertex in $\CG(G,R')[D\cup V^{+}(R')]$.
If in addition Conditions (1)-(3) are satisfied, we call $(D,(p,\wc(D)))$ a \emph{solution} in $\cD_G[R,R']$.

\medskip

It is straightforward to check that the weight of an optimum solution is the optimum over all $R\in \{0,\ldots,d\}^k$ of
$opt\{w\mid (\emptyset,w)\in \cD_G[R,\{0\}^k]\}$ for a $k$-labeled graph $G$. 

Analogously to Section \ref{sec:fvs} our dynamic programming algorithm will store a subset of $\cD_G[R,R']$ of size $2^{k-1}$ that represents $\cD_G[R,R']$.  Recall that we suppose that any graph is given with an irredundant $k$-expression.

\medskip

\paragraph{\bf Computing $tab_G$ for $G=\mathbf{1}(x)$} For $(r_1)\in \{0,\ldots,d\},(r'_1)\in \{0,\ldots,d\}$, let
{\small \begin{align*}
  tab_G[(r_1),(r_1')] &:=\begin{cases}
		  \{(\emptyset,0)\} & \textrm{if $r_1=0$ and $r_1'\in \rho$},\\ 
     \{(\{\{1\}\},\wc(x))\} & \textrm{if $r_1=1$ and $r_1'\in \sigma$,}\\
     \emptyset& \text{otherwise.}\end{cases}
\end{align*}}

Since there is only one vertex in $G$ labeled 1, $\cD_G[(r_1),(r_1')]$ is empty whenever $r_1\notin\{0,1\}$. 
Also, the possible solutions are either to put $x$ in the solution ($r_1=1$) or to discard it ($r_1=0$); in both cases we should check that $x$ is $(\sigma,\rho)$-dominated by $V^+_1(R')$.
We deduce then that $tab_G[(r_1),(r_1')]=\cD_G[(r_1),(r_1')]$.

\medskip
\paragraph{\bf Computing $tab_G$ for $G=ren_{i\to j}(H)$} We can suppose that $H$ is $k$-labeled and that $i=k$. 
Let $R=(r_1,\dots,r_{k-1}),R'=(r'_1,\dots,r'_{k-1})\in \{0,\ldots,d\}^{k-1}$.

To compute $tab_G[R,R']$, we define $\cS$, the set of tuples coherent with respect to $R$ and Condition (1), as follows
{\small \begin{align*}
	\cS:= \{ (s_1,\dots,s_k)\in \{0,\dots,d\}^k \mid  r_j = \min(d, s_k +s_j) \text{ and } \forall\ell \in [k]\setminus\{i,j\}, s_\ell=r_\ell \}.
	\end{align*}}
It is worth noticing that we always have $(r_1,\dots,r_{k-1},0)\in\cS$. Moreover, if $r_j=0$, then $\cS=\{ (r_1,\dots,r_{k-1},0)\}$.
We define also $S'=(s'_1,\dots,s_k')\in \{0,\dots,d\}^k$ with $s'_k=r'_j$ and $s_\ell' = r'_\ell$ for all $\ell \in [k-1]$. 
Notice that $S'$ is the only tuple compatible with $R'$ and Condition (2) since for every $v\in lab_G^{-1}(j)$,  the number of vertices in $V^{+}(R')$ adjacent to $v$ in $\CG(G,R')$ is the same as the number of vertices in $V^{+}(S')$ adjacent to $v$ in $\CG(H,S')$. 

\begin{enumerate}[(a)]
	
	\item If $r_j=0$ or $r_j'=0$, then we let \[ tab_G[R,R']:=\reduce\left(\rmc\left(\bigcup_{S\in\cS}tab_H[S,S']\right)\right). \]
	In this case, the vertices in $lab_G^{-1}(j)$ are not expected to play a role in the future as either we expect no neighbors for them in the future or they are not intersected by the partial solutions.
	
	\item  Otherwise, we let $tab_G[R,R']:=\reduce(\rmc(\cA))$ with
	\begin{align*}
		\cA:= \proj\left(\{k\}, \bigcup_{S\in\cS} \join ( tab_H[S,S'], \{ (\{ \{ j, k  \}\}, 0 ) \})\right).
	\end{align*}
	Intuitively, we put in $tab_G[R,R']$ all the weighted partitions $(p,w)$ from the tables $tab_H[S,S']$ with $S\in \cS$, after merging the blocks in $p$ containing $k$ and $j$, and removing $k$ from the resulting partition.
\end{enumerate}
	  	
\begin{lemma}\label{lem:rend} Let $G=ren_{k\to j}(H)$ with $H$ a $k$-labeled graph. For all $R=(r_1,\dots,r_{k-1}),R'=(r'_1,\dots,r_{k-1}')\in \{0,\dots,d\}^{k-1}$, the table $tab_G[R,R']$ is a representative set of  $\cD_G[R,R']$ assuming that $tab_H[S,S']$ is a representative set of $\cD_H[S,S']$ for all $S\in \{0,\ldots,d\}^k$.
\end{lemma}
\begin{proof} 
Since the used operators preserve representation, it is enough to prove that each weighted partition added to $tab_G[R,R']$ belongs to $\cD_G[R,R']$, and that
\begin{itemize}
	\item in Case (a), we have $\cD_G[R,R']\subseteq \bigcup_{S\in\cS}\cD_H[S,S']$, and
	\item in Case (b), we have $\cD_G[R,R']\subseteq \cA$ if we let $tab_H[S,S']=\cD_H[S,S']$ for every $S\in\{0,\dots,d\}^k$.
\end{itemize} 

Let $(D,(p,w))$  be a solution in $\cD_G[R,R']$. We start by proving that we have $(p,w)\in \bigcup_{S\in\cS}\cD_H[S,S']$ if we are in Case (a), or $(p,w)\in\cA$ if we are in Case (b).
By the definition of $\cS$, we deduce that $r_H(D)\in \cS$. Indeed, $r_{j,G}(D)=\min(d,| D \cap lab_G^{-1}(j)| )$ equals $\min(d, r^d_{j,H}(D)+ r^d_{k,H}(D))$ because $lab_G^{-1}(j)=lab_H^{-1}(\{j,k\})$.

Let $p'\in \Pi(\Active(r_H^d(D),S'))$ such that $(D,(p',w))$ is a candidate solution in $\cD_H[r_H^d(D),R']$.
We claim that $(D,(p',w))$ is a solution in $\cD_H[r_H^d(D),R']$.
Condition (1) is trivially satisfied. 
We deduce from the definition of $S'$ that Condition (2) is satisfied. 
We claim that Condition (3) is satisfied.
If $R'=\{0\}^{k-1}$, then we have $S'=\{0\}^{k}$ and Condition (3) is satisfied because $H[D]=G[D]$ must be connected since $(D,(p,w))$ is a solution in $\cD_G[R,R']$.
Otherwise, every connected component $C$ of $G[D]=H[D]$ intersects $lab_G^{-1}(\Active(R,R'))$.
Let $C$ be a connected component of $H[D]$.
If $C$ contains a vertex labeled $l$ in $G$ with $l\in\Active(R,R')\setminus \{ j\}$, then by definition of $S'$, we have $\ell \in\Active(r_H^d(D),S')$.
Suppose now, $C$ contains a vertex $v$ in $lab_G^{-1}(j)$ and $j\in \Active(R,R')$.
If $v$ is labeled $k$ in $H$, then $r_{k,H}^d(D)\ne0$ and thus $k$ belongs to $\Active(r_H^d(D),S')$ because $s'_k=r_j'\ne 0$.
Symmetrically, if $v$ is labeled $j$ in $H$, then $j\in \Active(r_H^d(D),S')$.
In both cases, $C$ intersects $lab_H^{-1}(\Active(r_H^d(D),S'))$.
We can conclude that Condition (3) is satisfied.
Hence, $(D,(p',w))$ is a solution in $\cD_H[r_H^d(D),R']$.

If $r_j=0$ or $r_j'=0$ (Case (a)), then it is easy to see that $p=p'$ from Equation (3) and thus $(p,w)\in tab_H[S,S']$.

Assume now that $r_j\ne 0$ and $r_j'\ne 0$ (Case (b)). 
Let $D_k:=D\cap lab_H^{-1}(k)$ and $D_j:=D\cap lab_H^{-1}(j)$.
Observe that the graph $\CG(G,R')[D\cup V^+(R')]$ can be obtained from the graph $\CG(H,S')[D\cup V^+(S')]$ by removing the vertices in $V^{+}_k(R')$ and by adding the edges between $V^+_j(R')$ and $D_k$.
Hence, $p$ is obtained from $p'$ by merging the blocks containing $j$ and $k$, and by removing $k$.
Thus, we can conclude that $(p,w)\in \cA$.

\medskip

It remains to prove that each weighted partition $(p,w)$ added to $tab_G[R,R']$ belongs to $\cD_G[R,R']$. 
Let $(p,w)$ be a weighted partition added to $tab_G[R,R']$ from $(p',w)\in tab_H[S,S']$, and let $D\subseteq V(H)$ such that $(D,(p',w))$ is a solution in $\cD_H[S,S']$.
We want to prove that $(D,(p,w))$ is a solution in $\cD_G[R,R']$.
From the definitions of $\cS$ and $S'$, we deduce that $D$ satisfies Conditions (1) and (2).
We deduce that $D$ satisfies also Condition (3) from Equation (3), the fact that $lab_G^{-1}(j)=lab_H^{-1}(\{j,k\})$, and because $G[D]=H[D]$. 

If $r_j=0$ or $r_j'=0$, then it is easy to see from Equation (3) that $p=p'$ and that $(D,(p,w))$ satisfies Condition (4).
Otherwise, we deduce from the previous observations concerning the differences between $\CG(G,R')[D\cup V^+(R')]$ and $\CG(H,S')[D\cup V^+(S')]$, that $(D,(p,w))$ satisfies Condition (4).
In both cases, we can conclude that $(D,(p,w))$ is a solution in $\cD_G[R,R']$.
\end{proof}

\medskip
\paragraph{\bf Computing $tab_G$ for $G=add_{i,j}(H)$} We can suppose that $H$ is $k$-labeled. Let $R=(r_1,\dots,r_k)\in \{0,\ldots,d\}^k,R'=(r'_1,\dots,r_k')\in \{0,\ldots,d\}^k$. 

Let $S':=(s'_1,\dots,s'_k)\in \{0,\dots,d\}^k$ such that $s_i':=\min(d,r'_i+r_j)$, $s_j':=\min(d,r'_j+r_i)$, and $s'_\ell = r'_\ell$ for all $\ell \in [k]\setminus \{i,j\}$. It is easy to see that $S'$ is the only tuple compatible with $R'$ and Condition (2). 
\begin{enumerate}[(a)]
	\item 	If $\Active(R,R')=\emptyset$, then we let \[ tab_G[R,R']:=\rmc\left( \{ (\emptyset, w ) \mid (p,w) \in tab_H[R,S'] \} \right). \] 
	In this case, the partial solution in $tab_G[R,R']$ are associated with connected solutions by Condition (3). 
	The partial solutions in $tab_H[R,S']$ trivially satisfy this condition in $G$. 
	Notice that $tab_G[R,R']$ represents $\cD_G[R,R']$ because the function $f:2^{\Pi(V)\times \bN}\to 2^{\{\emptyset\}\times \bN}$ with $f(\cA):=\{(\emptyset,w)\mid (p,w)\in\cA\}$ preserves representation.
	
	\item If $r_i=0$ or $r_j=0$, we let $tab_G[R,R']:= tab_H[R,S']$. We just copy all the solutions not intersecting $lab_G^{-1}(i)$ or $lab_G^{-1}(j)$. In this case, the connectivity of the solutions is not affected by the $add_{i,j}$ operation.
	
	\item Otherwise, we let $tab_G[R,R']:=\rmc(\cA)$, where
	\begin{align*}
		\cA:=\proj(\{t \in \{ i,j\}\mid r_t'=0\}, \join(tab_H[R,S'], \{ (\{ \{ i,j\}\}, 0) \})).
	\end{align*}
	In this last case, we have $i,j\in\Active(R,S')$. We put in $tab_G[R,R']$, the weighted partitions of 
	$tab_H[R,S']$ after merging the blocks containing $i$ and $j$, and removing $i$ or $j$ if, respectively, $r_i'=0$ and $r_j'=0$, \ie, if they don't belong, respectively, to $\Active(R,R')$. 
\end{enumerate}
It is worth noticing that if $|tab_H[R,S']|\leq 2^{k-1}$, then we have $|tab_G[R,S']|\leq 2^{k-1}$. Thus, we do not have to use the function $\reduce$ to compute $tab_G$.

\begin{lemma}\label{lem:addd} Let $G=add_{i,j}(H)$ be a $k$-labeled graph. For all tuples $R=(r_1,\dots,r_k),R'=(r'_1,\dots,r_k')\in \{0,\ldots,d\}^k$, the table $tab_G[R,R']$ is a representative set of  $\cD_G[R,R']$ assuming that $tab_H[R,S']$ is a representative set of $\cD_H[R,S']$.
\end{lemma}

\begin{proof} 
	Since the used operators preserve representation, it is enough  to prove that every weighted partition added to $tab_G[R,R']$ belongs to $\cD_G[R,R']$, and that 
\begin{itemize}
	\item in Case (a), we have $\cD_G[R,R']\subseteq \{ (\emptyset, w ) \mid (p,w) \in \cD_H[R,S'] \}$,
	\item  in Case (b), we have $\cD_G[R,R']\subseteq\cD_H[R,S']$, and
	\item in Case (c), we have  $\cD_G[R,R']\subseteq\cA$ if we let $tab_H[R,S']=\cD_H[R,S']$.
\end{itemize}
	
	Let $(D,(p,w))$ be a solution in $\cD_G[R,R']$.
	Let $p'\in \Pi(\Active(R,S'))$ such that $(D,(p',w))$ is a candidate solution in $\cD_H[R,S']$.
	We claim that $(D,(p',w))$ is a solution in $\cD_H[R,S']$.
	Condition (1) is trivially satisfied because $lab_G(v)=lab_H(v)$ for all $v\in V(G)$.
	We claim that Condition (2)  is satisfied, that is $D\cup V^+(S')$ $(\sigma,\rho)$ dominates $V(H)$ in $\CG(H,S')$.
	It is quite easy to see that $D\cup V^+(S')$ $(\sigma,\rho)$ dominates $V(H)\setminus lab_H^{-1}(\{i,j\})$.
	Let $v\in lab_H^{-1}(i)$. We claim that $v$ is $(\sigma,\rho)$ dominated by $D\cup V^+(S')$.
	Because we consider only irredundant $k$-expressions, there is no edge between an $i$-vertex and a $j$-vertex in $H$. 
	Therefore, we have $|N_G(v)\cap D|= |N_H(v)\cap D| + |D\cap lab_G^{-1}(j)|$.
	Since $D\cup V^+(R')$ $(\sigma,\rho)$-dominates $V(G)$, if $v\in D$, then $|N_G(v)\cap D| + r_i'\in \sigma$, otherwise, $|N_G(v)\cap D| + r_i'\in\rho$.
	By Fact \ref{fact:d}, we conclude that $\min(d, |N_H(v)\cap D| + |D\cap lab_G^{-1}(j)| +r_i' )$ belongs to $\sigma$ if $v\in D$, otherwise it belongs to $\rho$.
	As $\min(d, r_i' + |D\cap lab_G^{-1}(j)|) =\min(d, r_i' + r_j)=s_i'$, we deduce that $D\cup V^+(S')$ $(\sigma,\rho)$ dominates $v$.
	Thus, every $i$-vertex is $(\sigma,\rho)$-dominated by $D\cup V^+(S')$ in $\CG(H,S')$.
	Symmetrically, we deduce that every $j$-vertex is $(\sigma,\rho)$-dominated by $D\cup V^+(S')$ in $\CG(H,S')$.
	Hence, Condition (2) is satisfied.
	
	In order to prove that Condition (3) is satisfied, we distinguish the following cases.
	First, suppose that $\Active(R,R')=\emptyset$. 
	By Condition (3), $G[D]$ is connected.
	As $G=add_{i,j}(H)$, the graph $H[D]$ is obtained from $G[D]$ by removing all edges between the $i$-vertices and the $j$-vertices.
	We deduce that either $H[D]$ is connected (if $r_i=0$ or $r_j=0$), or that every connected component of $H[D]$ contains at least one vertex whose label is $i$ or $j$ (otherwise $G[D]$ would not be connected). 
	In both cases, Condition (3) is satisfied.
	
	Assume now that $\Active(R,R')\ne\emptyset$. 
	If $r_i=0$ (resp. $r_j=0$), then, by definition of $S'$, we have $s_j'=r_j'$ (resp. $s_i'=r_i'$).
	Therefore, if $r_i=0$ or $r_j=0$, then we have $\Active(R,R')=\Active(R,S')$, and Condition (3) is trivially satisfied.
	Otherwise, if $r_i\ne0$ and $r_j\ne0$, then, by definition of $S'$, we have $s_i'\ne 0$, $s_j'\ne 0$, and thus $i,j\in \Active(R,S')$. 
	In this case, we conclude that Condition (3) is satisfied because, for every connected component $C$ of $H[D]$, either $C$ is a connected component of $G[D]$, or $C$ contains at least one vertex whose label is $i$ or $j$.
	
	Hence, $(D,(p',w))$ is a solution in $\cD_H[R,S']$.
	If $\Active(R,R')=\emptyset$, then $p=\emptyset$ and $(p,w)$ is added in $tab_G[R,R']$.
	Else if $r_i=0$ or $r_j=0$, then $p'=p$ and $(p,w)$ is also added to $tab_G[R,R']$.
	Otherwise, it is easy to see that $p$ is obtained from $p'$ by merging the blocks of $p'$ containing $i$ and $j$, and by removing them if they belong to $\{\ell \in [k]\mid r'_\ell \ne 0\}$. Thus, we can conclude that $(p,w)\in \{ (\emptyset, w ) \mid (p,w) \in \cD_H[R,S'] \} )$ in Case (a), $(p,w)\in\cD_H[R,S']$ in Case (b), and $(p,w)\in \cA$ in Case (c).
	\medskip
	
	It remains to prove that every weighted partition added to $tab_G[R,R']$ belongs to $\cD_G[R,R']$.
	Let $(p,w)$ be a weighted partition added to $tab_G[R,R']$ from $(p',w)\in tab_H[R,S']$ and let $D\subseteq V(H)$ such that $(D,(p',w))$ is a solution in $\cD_H[R,S']$.
	We claim that $(D,(p,w))$ is a solution in $\cD_G[R,R']$. 
	Obviously, Condition (1) is satisfied.
	We deduce that Condition (2) is satisfied from the definition of $S'$ and the previous arguments.
	We claim that Condition (3) is satisfied. 
	Suppose first that $R'=\{0\}^k$. Then either $S'=\{0\}^k$ or $\{i,j\}=\Active(R,S')$.  
	Indeed, we have $\Active(R,S')\subseteq \{i,j\}$. 
	Now, $i\in\Active(R,S')$ if and only if $r_i\ne 0$ and $s_i'\ne 0$, and then, by definition of $S'$, $s'_j= r_i\ne 0$ and $s_i'=r_j\ne 0$.
	Thus, $i\in\Active(R,S')$ if and only if $j\in\Active(R,S')$.
	Therefore, we conclude that either $H[D]$ is connected or every connected component $C$ of $H[D]$ contains at least a vertex whose label is $i$ or $j$. As $G[D]$ is obtained from $H[D]$ by adding all the edges between the $i$-vertices and the $j$-vertices, we conclude that, in both cases, $G[D]$ is connected.
	
	Otherwise, if $R'\ne\{0\}^k$, then every connected component of $H[D]$ contains at least one vertex whose label is in $\Active(R,S')$.
	If $r_i=0$ or $r_j=0$, then we are done because $\Active(R,R')=\Active(R,S')$ by definition of $S'$.
	Suppose now that $r_i\neq0$ and $r_j\ne 0$. 
	In this case, we have $i,j\in\Active(R,S')$ from the definition of $S'$. 
	Assume towards a contradiction that there exists a connected component $C$ in $G[D]$ such that $C$ does not intersect $lab_G^{-1}(\Active(R,R'))$. Notice that $C$ intersects $lab_G^{-1}(\Active(R,S'))$ because $C$ is a union of connected components of $H[D]$. 
	As $\Active(R,S')\setminus \Active(R,R')\subseteq \{i,j\}$, we deduce that $C$ contains at least one vertex whose label is $i$ or $j$. 	
	Since the $i$-vertices and the $j$-vertices of $D$ are in the same connected component of $G[D]$, we have $lab_G^{-1}(\{i,j\})\cap D\subseteq C$. 
	Therefore, we conclude that $\Active(R,S')\setminus \Active(R,R') = \{i,j\}$. 
	It follows, that all the connected components of $H[D]$ that intersect $lab_G^{-1}(\{i,j\})$ does not intersect $lab_G^{-1}(\Active(R,S')\setminus \{i,j\})$ because they are contained in $C$. We can conclude that $\{i,j\}$ is a block of $p'' :=p'\sqcup \{ \{ i,j\}\}_{\uparrow \Active(R,S')}$.
	Hence, we have $\proj(\{t \in \{ i,j\}\mid r_t'=0\}, \{ (p'',w) \})=\emptyset$ because $\{t \in \{ i,j\}\mid r_t'=0\}=\{i,j\}$. This contradicts the fact that $(p,w)$ is obtained from $(p',w')$. Thus, Condition (3) is satisfied.
	
	We deduce from the previous observations concerning Condition (4) that this condition is also satisfied. Thus, every solution $(p,w)$ added to $tab_G[R,R']$ belongs to $\cD_G[R,R']$.
\end{proof}

\medskip
\paragraph{\bf Computing $tab_G$ for $G=G_a\oplus G_b$} We can suppose w.l.o.g. that $G_a$ and $G_b$ are $k$-labeled\footnote{For example, if the set of labels of $G_a$ (or $G_b$) is $[k]\setminus 1$, then we can extend any tuples $R=(r_2,\dots,r_k),R'=(r'_2,\dots,r_k')\in \{0,\ldots,d\}^{k-1}$, to $\{0,\ldots,d\}^{k}$ by adding $r_1=0$ to $R$ and $r_1'=*$ to $R'$ with $*$ a special value considered as equal to all the integers in $\{0,\dots,d\}$.}. 
Let $R=(r_1,\dots,r_k),R'=(r'_1,\dots,r_k')\in \{0,\ldots,d\}^k$. 

Let $A=(a_1,\dots,a_k),B=(b_1,\dots,b_k)\in \{0,\ldots,d\}^k$. The following notion characterizes the pairs $(A,B)$ compatible with $R$ with respect to Condition (1).
We say that $(A,B)$ is $R$\emph{-compatible} if and only if for all $i\in [k]$, we have $r_i=\min( d , a_i + b_i)$.

\begin{enumerate}[(a)]
	\item If $R'=\{0\}^k$, then we let $tab_G[R,R']:=\reduce(\rmc( tab_{G_a}[R,R'] \cup tab_{G_b}[R,R']  ))$ if $0\in \rho$, otherwise we let $tab_G[R,R']=\emptyset$.
	Condition (3) implies that the partial solutions in $\cD_G[R,R']$ are either fully contained in $V(G_a)$ or in $V(G_b)$ since there are no edges between these vertex sets in $G$.
	Moreover, in order to satisfy Condition (2), we must have $0\in\sigma$.
	
	\item Otherwise, we let $tab_G[R,R']:=\reduce(\rmc(\cA))$ where
	\begin{align*}
	\cA:= \bigcup_{(A,B) \text{ is } R\text{-compatible}} \join( tab_{G_a}[A,R'] , tab_{G_b}[B,R'] ).
	\end{align*}
\end{enumerate}

\begin{lemma}\label{lem:joind} 
	Let $G=G_a\oplus G_b$ be a $k$-labeled graph. For all $R=(r_1,\dots,r_k),R'=(r'_1,\dots,r_k')\in \{0,\ldots,d\}^k$, the table $tab_G[R,R']$ is a representative set of $\cD_G[R,R']$ assuming that $tab_{G_a}[A,R']$ and $tab_{G_a}[B,R']$ are representative sets of $\cD_{G_a}[A,R']$ and $\cD_{G_b}[B,R']$, respectively, for all $A,B\in \{0,\ldots,d\}^k$. 
\end{lemma}

\begin{proof}
		Since the used operators preserve representation, it is easy to see that if $R'=\{0\}^k$, then we are done as $\cD_G[R,R']=\cD_{G_a}[R,R']\cup \cD_{G_b}[R,R']$ if $0\in\rho$, otherwise $\cD_G[R,R']=\emptyset$. 
		Indeed, by Condition (3), for all solutions $(D,(p,w))$ in $\cD_G[R,R']$, the graph $G[D]$ must be connected. 
		Since $G=G_a\oplus G_b$, there are no edges between $V(G_a)$ and $V(G_b)$ in $G[D]$. 
		Thus, $D$ is either included in $V(G_a)$ or in $V(G_b)$. 
		Since $V(G_a)\ne\emptyset$ and $V(G_b)\ne\emptyset$, we have $V(G)\setminus D\ne \emptyset$.
		This implies that $0\in\rho$, otherwise the vertices in $V(G)\setminus D$ will not be $(\sigma,\rho)$-dominated by $D\cup V^+(R')=D$ in $\CG(G,R')$ as $V^+(R')=\emptyset$.
		
		In the following, we assume that $R'\neq \{0\}^k$. Since the used operators preserve representation, it is enough to prove that $\cD_G[R,R']=\cA$ if we let $tab_{G_t}[S,R']=\cD_{G_t}[S,R']$ for all $S\in\{0,\dots,d\}^k$ and $t\in\{a,b\}$.
		
		Let $(D,(p,w))$ be a solution in $\cD_G[R,R']$. We start by proving that $(p,w)\in\cA$.
		Let $D_a=D\cap V(G_a)$ and $D_b=D\cap V(G_b)$.
		From the definition of $R$-compatibility, we deduce that $r^d_{G_a}(D_a)$ and $r^d_{G_b}(D_b)$ are $R$-compatible.
		Indeed, we have $\min(d, |lab_G^{-1}(i)\cap D|)= \min( d , r^d_{G_a}(D_a) + r^d_{G_b}(D_b) )$ for all $i\in[k]$.
		
		Let $p_a\in\Pi(\Active(r_{G_a}^d(D_a),R'))$ such that $(D_a,(p_a,\wc(D_a)))$  is a candidate solution in $\cD_{G_a}[r_{G_a}^d(D_a),R']$.  
		We claim that the pair $(D_a,(p_a,\wc(D_a)))$ is a solution in $\cD_{G_a}[r_{G_a}^d(D_a),R']$.
		Condition (1) is trivially satisfied.
		By assumption, $D\cup V^+(R')$ $(\sigma,\rho)$ dominates $V(G)$ in $\CG(G,R')$.
		Since there are no edges between the vertices in $V(G_a)$ and those in $V(G_b)$, we conclude that $D_a\cup V^{+}(R')$ $(\sigma,\rho)$-dominates $V(G_a)$ in $\CG(G_a,R')$. That is Condition (2) is satisfied.
		Observe that every connected component of $G[D]$ is either included in $D_a$ or in $D_b$. 
		Therefore, every connected component of $G_a[D_a]$ contains a vertex $v$ with a label $j\in\Active(R,R')$.
		Since $v\in D_a$, we conclude that $r_{j,G_a}^d(D_a)\ne0$, thus $j\in \Active(r^d_{G_a}(D_a),R')$. 
		We can conclude that Condition (3) is satisfied.
		Thus, $(D_a,(p_a,\wc(D_a)))$ is a solution in $\cD_{G_a}[r_{G_a}^d(D_a),R']$.
		Symmetrically, we deduce that there exits $p_b\in \Active(R_{G_b}^d(D_b),R')$ such that  $(D_b,(p_b,\wc(D_b)))$ is a solution in $\cD_{G_b}[r_{G_b}^d(D_b),R']$.
		
		It remains to prove that $p=p_{a\uparrow V}\sqcup p_{b\uparrow V}$ with $V=\Active(R,R')$. First, observe that $\Active(r^d_{G_a}(D_a),R')$ and $\Active(r^d_{G_b}(D_b,R'))$ are both subset of $\Active(R,R')$ and thus $p_{a\uparrow V}\sqcup p_{b\uparrow V}$ is a partition of $\Active(R,R')$.
		Let $\sim_a$ (resp. $\sim_b$) be the equivalence relation such that $i\sim_a j$ (resp. $i\sim_b j$) if an $i$-vertex is connected to a $j$-vertex in the graph $\CG(G_a,R')[D_a\cup V^{+}(R')]$ (resp. $\CG(G_b,R')[D_b\cup V^+(R')])$.
		By Condition (4), two labels $i,j$ are in the same block of $p$ if and only if an $i$-vertex and a $j$-vertex are connected in $\CG(G,R')[D\cup V^{+}(R')]$. On the other hand, $i$ and $j$ are in the same block of $p_{a\uparrow V}\sqcup p_{b\uparrow V}$ if and only if $i\cR j$ where $\cR$ is the transitive closure of the relation ($i\sim_a j$ or $i\sim_b j$). 
		By definition of $\CG(G,R')$, for every label $i\in \Active(R,R')$, the vertices in $lab_G^{-1}(i)\cap D$ are all adjacent to the vertices in $V^+_i(R')$.
		One can easily deduce from these observations that $p=p_{a\uparrow V}\sqcup p_{b\uparrow V}$. Hence, $(p,w)\in\cA$.
		
		\medskip
		We now prove that every weighted partition in $\cA$ belongs to $\cD_G[R,R']$.
		Let $(p,w)$ be a weighted partition added to $tab_G[R,R']$ from a solution $(D_a,(p_a,w_a))$ in $\cD_{G_a}[A,R']$ and a solution $(D_b,(p_b,w_b))$ in $\cD_{G_b}[B,R']$. 
		We claim that $(D_a\cup D_b,(p_{a\uparrow V}\sqcup p_{b\uparrow V}, w_a + w_b))$ is a solution in $\cD_G[R,R']$ with $V=\Active(R,R')$.
		We deduce that Condition (1) is satisfied from the definition of $R$-compatibility and because $\min(d, |lab_G^{-1}(i)\cap D|)=\min( d , r^d_{G_a}(D_a) + r^d_{G_b}(D_b) )$ for all $i\in [k]$.
		With the same arguments given previously, one easily deduces that Conditions (2)-(4) are also satisfied. 
		We conclude that $(D_a\cup D_b,(p_{a\uparrow V}\sqcup p_{b\uparrow V}, w_a + w_b))$ is a solution in $\cD_G[R,R']$.
\end{proof}

\begin{theorem}\label{thm:dom} There is an algorithm that, given an $n$-vertex graph $G$ and an irredundant $k$-expression of $G$, computes a maximum (or a minimum) connected $(\sigma,\rho)$-dominating set in time $(d+1)^{3k}\cdot 2^{(\omega+1)\cdot k}\cdot k^{O(1)}\cdot n$ with $d=\max(d(\sigma),d(\rho))$.
\end{theorem}

\begin{proof} We do a bottom-up traversal of the $k$-expression and at each step we update the tables as indicated above. The correctness of the algorithm  follows from Lemmas \ref{lem:rend}-\ref{lem:joind}. 
	From the definition of $\cD_G[R,R']$, we deduce that the weight of an optimum connected $(\sigma, \rho)$-dominating set corresponds to the optimum over all $R\in\{0,\dots,d\}^k$ of $\opt\{w \mid (\emptyset, w)\in tab_G[R,\{0\}^k]\}$ because $tab_G[R,\{0\}^k]$ represents $\cD_G[R,\{0\}^k]$.
	
	Let us discuss the time complexity now. We claim that the tables of $tab_G$ can be computed in time $(d+1)^{3k}\cdot 2^{(\omega+1)\cdot k}\cdot k^{O(1)}$.
	We distinguish the following cases:
	\begin{itemize}
		\item If $G=\mathbf{1}(x)$, then it is easy to see that $tab_G$ is computable in time $O(d)$.
		\item If $G=add_{i,j}(H)$, then we update $tab_G[R,R']$ from one entry $tab_H[R,S']$ for some fixed $S'$ computable in constant time. The used $\join$ operation runs in time $2^{k-1}\cdot k^{O(1)}$ (from Fact \ref{fact:join}). Thus, $tab_G$ is computable in time $(d+1)^{2k}\cdot  k^{O(1)}$. 
		\item Now, if $G=ren_{i\to j}(H)$, then we update $tab_G[R,R']$ from at most $|\cS|=(d+1)^2$ tables from $tab_H$, each identified in constant time from $(R,R')$.  Since each table of $tab_H$ contains at most $2^{k-1}$ entries, computing the call at the function $\reduce$ take $(d+1)^2\cdot  2^{\omega\cdot k}\cdot k^{O(1)}$. Thus, we can compute $tab_G$ in time $(d+1)^{2k+2}\cdot 2^{ \omega\cdot k}\cdot k^{O(1)}$.  
		\item If $G=G_a\oplus G_b$, then the bottleneck is when $R'\neq \{0\}^k$. Indeed, if $R'= \{0\}^k$, then $tab_G[R,R']$ can be computed in time $O(2^{\omega\cdot k})$ since $tab_G[R,R']$ is computed from two tables, each containing at most $2^{k-1}$ entries.
		Let $R'\neq \{0\}^k$. 
		By Theorem \ref{thm:reduce2}, we can compute the tables $tab_G[R,R']$, for every $R\in \{1,\dots,d\}^k$ in time 
		{\small \begin{align*}
			\sum_{R\in\{0,\dots,d\}^k} \left(\sum_{\substack{(A,B) \text{ is }\\ R\text{-comptatible}} }|\join(tab_{G_a}[A,R'], tab_{G_b}[B,R'])|\cdot 2^{(\omega-1)\cdot k}\cdot k^{O(1)}\right).
			\end{align*}}
		Observe that for all $A,B\in \{0,\dots,d\}^k$:
		\begin{enumerate}
			\item There is only one $R\in \{0,\dots,d\}^k$ such that $(A,B)$ is $R$-compatible. This follows from the definition of $R$-compatibility. Hence, there are at most $(d+1)^{2k}$ tuples $(A,B,R)$ such that $(A,B)$ is $R$-compatible.
			\item The size of $\join(tab_{G_a}[A,R'], tab_{G_b}[B,R'])$ is bounded by $2^{2(k-1)}$ and this set can be computed in time $2^{2(k-1)}\cdot k^{O(1)}$.
		\end{enumerate}
		Since $2^{2(k-1)}\cdot 2^{(\omega-1)\cdot k}\leq 2^{(\omega +1)\cdot k}$, we conclude from Observations (1)-(2) that we can compute the tables $tab_G[R,R']$, for every $R\in \{1,\dots,d\}^k$, in time $(d+1)^{2k} \cdot 2^{(\omega + 1 )\cdot k}\cdot k^{O(1)}$.
		
		Hence, we can update $tab_G$ in time $(d+1)^{3k} \cdot 2^{(\omega + 1 )\cdot k}\cdot k^{O(1)}$.
	\end{itemize}
  Hence, in the worst case, the tables of $tab_G$ takes $(d+1)^{3k} \cdot 2^{(\omega + 1 )\cdot k}\cdot k^{O(1)}$ time to be computed.
 Because the size of a $k$-expression is $O(n\cdot k^2)$, we can conclude that a maximum (or minimum) $(\sigma,\rho)$-dominating set can be computed in the given time.
\end{proof}

As a consequence of Theorem \ref{thm:dom}, we have the following corollary.

\begin{corollary}\label{cor:st} There is an algorithm that, given an $n$-vertex graph $G$, a subset $K\subseteq V(G)$ and an irredundant $k$-expression of $G$, computes a minimum node-weighted Steiner tree for $(G,K)$ in time $2^{(\omega+4)\cdot k}\cdot k^{O(1)}\cdot n$.
\end{corollary}
\begin{proof} 
	We can assume w.l.o.g. that $|K|\geq 2$. 
	We can reduce the problem \textsc{Node-weighted Steiner Tree} to a variant of $(\sigma,\rho)$\textsc{-Dominating Set} where $\sigma=\bN^+$ and $\rho=\bN$. This variant requires $K$ to be included in the $(\sigma,\rho)$-dominating set. 
	We can add this constraint, by modifying how we compute the table $tab_G$, when $G=\mathbf{1}(x)$ and $x\in K$. For $(r_1),(r_1')\in \{0,1\}$, we let
	{\small \begin{align*}
		tab_G[R,R'] &:=\begin{cases} 
			\{(\{\{1\}\},\wc(x))\} & \textrm{if $r_1= 1$ and $r_1'= 1$},\\ 
			 \emptyset &\textrm{otherwise.}\\ \end{cases}
		\end{align*}}
	It is straightforward to check that this modification implements this constraint and our algorithm with this modification computes a minimum node-weighted Steiner tree. The running time follows from the running time of Theorem \ref{thm:dom} with $d=1$.
\end{proof}

More modifications are needed in order to compute a maximum (or minimum) connected co-$(\sigma,\rho)$-dominating set.

\begin{corollary}\label{cor:outDS}
	There is an algorithm that, given an $n$-vertex graph $G$ and an irredundant $k$-expression of $G$, computes a maximum (or minimum) co-$(\sigma,\rho)$-domina\-ting set in time $(d+2)^{3k}\cdot 2^{(\omega +1)\cdot k} \cdot k^{O(1)} \cdot n$.
\end{corollary}
\begin{proof}
	First, we need to modify the definition of the tables $\cD_G$. Let $H$ be a $k$-labeled graph, $R,R'\in\{0,\dots,d\}^k$, and $\comp{R},\comp{R}'\in\{0,1\}^k$. The entries of $\cD_H[R,R',\comp{R},\comp{R}']$ are all the weighted partitions $(p,w)\in \Pi(\Active(\comp{R},\comp{R}'))\times \bN$ such that there exists a set $X\subseteq V(G)$  so that $w=\wc(X)$ and 
	\begin{enumerate}
		\item $r^d_{H}(V(H)\setminus X)=R$ and $r^1_{H}(X)=\comp{R}$,
		\item $(V(H)\setminus X)\cup V^+(R')$ $(\sigma,\rho)$-dominates $V(H)$ in $\CG(H,R')$,
		\item if $\Active(\comp{R},\comp{R}')=\emptyset$, then $H[X]$ is connected, otherwise every connected component of $H[X]$ intersects $lab_H^{-1}(\Active(\comp{R},\comp{R}'))$,
		\item $p= \Active(\comp{R},\comp{R}')/\sim$ where $i\sim j$ if and only if an $i$-vertex is connected to a $j$-vertex in $\CG(H,\comp{R}')[X\cup V^+(\comp{R}')]$.
	\end{enumerate}
As a solution is a set $X$ such that $V(G)\setminus X$ is a $(\sigma,\rho)$-dominating set and $G[X]$ is a connected graph, we need information about $X\cap V(H)$ and $(V(H)\setminus X)$.
Intuitively, $R,R'$ are the information we need to guarantee the $(\sigma,\rho)$-domination, and $\comp{R},\comp{R}'$ are the information we need to guarantee the connectedness. In particular, $\comp{R}$ specifies which label classes are intersected and $\comp{R}'$ tells which label classes are expected to have at least one neighbor in the future.

These modifications imply in particular to change the notion of $R$-compatibility. For each $t\in\{a,b,c\}$, let $R_t=(r^t_1,\dots,r^t_k)\in\{0,\dots,d\}^k$, and $\comp{R_t}=(\comp{r}_1^t,\dots,\comp{r}^t_k)\in\{0,1\}^k$, we say that $(R_a,\comp{R_a},R_b,\comp{R_b})$ is $(R_c,\comp{R_c})$-compatible if for all $i\in[k]$, we have $r_i^c=\min(d, r_i^a + r_i^b)$ and $\comp{r}_i^c=\min(1,\comp{r}^a_i + \comp{r}_i^b)$.

It is now an exercise to modify the algorithm of Theorem \ref{thm:dom} in order to update the tables $tab_H$ through the different clique-width operations.
The weight of an optimum solution corresponds to the optimum over all $R\in \{0,\dots,d\}^k,\comp{R}\in \{0,1\}^k$ of $\opt\{ w \mid (\emptyset, w) \in tab_G[R,\{0\}^k,\comp{R},\{0\}^k]\}$, since $tab_G[R,\{0\}^k,\comp{R},\{0\}^k]$ represents $\cD_G[R,\{0\}^k,\comp{R},\{0\}^k]$.

Let us discuss the time complexity now. Let $H$ be a $k$-labeled graph that is used in the $k$-expression of $G$. 
First, observe that we do not need to compute $tab_H[R,R',\comp{R},\comp{R}']$ for all $R,R'\in\{0,\dots,d\}^k$, and $\comp{R},\comp{R}'\in\{0,1\}^k$. 
Indeed, for all $X\subseteq V(H)$ and $i\in [k]$, if we have $r_{i,H}^d(V(H)\setminus X)< d$, then 
\begin{align*}
r_{i,H}^1(X)=\begin{cases}
0 & \text{if $|lab_H^{-1}(i)|= r_{i,H}^d(V(H)\setminus X)$},\\
1 & \text{otherwise}.
\end{cases}
\end{align*}
Hence, we deduce that there are at most $(d+2)^k$ pairs $(R,\comp{R})\in \{0,\dots,d\}^k\times\{0,1\}^k$ such that $R=r^d_H(V(H)\setminus X)$ and $\comp{R}=r^1_H(X)$.
Indeed, whenever $r_{i,H}^d(V(H)\setminus X)< d$, there is only one possible value for $r_{i,H}^1(X)$, and when $r_{i,H}^d(V(H)\setminus X)=d$, there are at most 2 possible values for $r_{i,H}^1(X)$.
With the same arguments used to prove the running time of Theorem \ref{thm:dom}, one easily deduces that there are at most $(d+2)^{2k}$ tuples $(R_a,\comp{R_a},R_b,\comp{R_b},R_c,\comp{R_c})$ such that $(R_a,\comp{R_a},R_b,\comp{R_b})$ is $(R_c,\comp{R}_c)$-compatible.

Moreover, it is sufficient to consider $(d+2)^k$ pairs $(R',\comp{R}')\in\{0,\dots,d\}^k\times\{0,1\}^k$ when we update $tab_H$.
For every $i\in [k]$, we let $c_H^i:=|N_G(lab_H^{-1}(i))\setminus N_H(lab_H^{-1}(i))|$.
Notice that for every vertex $v\in lab_H^{-1}(i)$, we have $|N_G(v)|=|N_H(v)| +c_H^i$. 
Informally, we cannot expect more than $c_H^i$ neighbors in the future for the $i$-vertices of $H$. 
Hence, it is enough to consider the pairs $(R',\comp{R}')\in\{0,\dots,d\}^k\times\{0,1\}^k$, with $R'=(r_1',\dots,r_k')$ and $\comp{R}'=(\comp{r}'_1,\dots,\comp{r}'_k)$, such that for all $i\in[k]$, if $r'_i < d$, then 
\begin{align*}
\comp{r}_i'=\begin{cases}
0 & \text{if $r_i'\geq c_i^H$},\\
1 & \text{otherwise}.
\end{cases}
\end{align*}
That is, for every $i\in[k]$, if $r_i'< d$, then there is one possible value for $\comp{r}'_i$ because if we expect $r_i'<d$ neighbors for the $i$-vertices in $V(H)\setminus X$, then we must expect $\min(0,c_H^i-r_i')$ neighbors for the $i$-vertices in $X$. 
If $r_i'=d$, then there are no restrictions on the value of $\comp{r}_i'$. 
Thus, the pairs $(r_i',\comp{r}_i')$ can take up to $(d+2)$ values.
We conclude that there are at most $(d+2)^k$ pairs $(R',\comp{R}')\in\{0,\dots,d\}^k\times\{0,1\}^k$ worth to looking at.
With these observations and the arguments used in the running time proof of Theorem \ref{thm:dom}, we conclude that we can compute a maximum (or minimum) co-$(\sigma,\rho)$-dominating set in the given time.
\end{proof}

\section{Concluding Remarks}

We combine the techniques introduced in \cite{BuiXuanTV13} and the rank-based approach from \cite{BodlaenderCKN15} to obtain $2^{O(k)}\cdot n$ time algorithms for several connectivity constraints problems such as \textsc{Connected Dominating Set}, \textsc{Node-Weighted Steiner Tree}, \textsc{Feedback Vertex Set}, \textsc{Connected Vertex Cover}, etc. 
While we did not consider connectivity constraints on locally vertex partitioning problems \cite{BuiXuanTV13}, it seems clear that we can adapt our algorithms from the paper to consider such connectivity constraints: if the solution is $\{D_1,\ldots,D_q\}$, each block $D_i$ is connected or a proper subset of the blocks form a connected graph.  
We did not consider counting versions and it would be interesting to know if we can adapt the approach in \cite{BodlaenderCKN15} based on the determinant to the clique-width.

In \cite{FominLPS16} Fomin et al. use fast computation of representative sets in matroids to provide deterministic $2^{O(k)}\cdot n^{O(1)}$ time algorithms parameterized by tree-width for many connectivity problems. Is this approach also generalizable to clique-width ?

The main drawback of clique-width is that, for fixed $k$, there is no known FPT polynomial time algorithm that produces a $k$-expression, that even approximates within a constant factor the clique-width of the given graph. 
We can avoid this major open question, by using other parameters as powerful as clique-width. Oum and Seymour introduced the notion of \emph{rank-width} and its associated \emph{rank-decomposition} \cite{OumS06} such that $\mathsf{rank\textrm{-}width}(G) \leq \mathsf{clique\textrm{-}width}(G)\leq 2^{\mathsf{rank\textrm{-}width}(G)+1}-1$. Furthermore there is a $2^{O(k)}\cdot n^3$ time algorithm for computing it \cite{HlinenyO08,Oum09a}. 
If we use our approach with a rank-decomposition, the number of twin-classes is bounded by $2^k$, and so we will get a $2^{2^{O(k)}}\cdot n^{O(1)}$ time algorithm. 
We can circle it and obtain $2^{O(k^2)}\cdot n^{O(1)}$ as done for example in \cite{GanianH10} for \textsc{Feedback Vertex Set} by using Myhill-Nerode congruences.

In \cite{BergougnouxK18}, the authors adapted the rank-based approach to the notion of \emph{$d$-neighbor equivalence} from \cite{BuiXuanTV13}.
They deduce, in particular, $2^{O(k)}\cdot n^{O(1)}$, $2^{O(k\cdot \log(k))}\cdot n^{O(1)}$, and $2^{O(k^2)}\cdot n^{O(1)}$ time algorithms parameterized respectively by clique-width, $\mathbb{Q}$-rank-width (a variant of rank-width \cite{OumSV13}), and rank-width for all the problems considered in this paper. 
The results in \cite{BergougnouxK18} generalize, simplify, and unify several results including \cite{BodlaenderCKN15,GanianH10} and those from this paper.
However, for each considered problem in this article, the algorithms parameterized by clique-width from \cite{BergougnouxK18} have a worse running time than the algorithms from this paper. 

We recall that, it is still open whether we can obtain a $2^{O(k)}\cdot n^{O(1)}$ time algorithm parameterized by rank-with (or $\mathbb{Q}$-rank-width) to solve the problems considered in this paper, or more basic problems such as \textsc{Independent Set}  or \textsc{Dominating Set}.

\bibliographystyle{plain}
\bibliography{biblio}

\end{document}